\pdfoutput=1
\documentclass[11pt]{article}
\usepackage{graphicx}
\usepackage[margin=1in]{geometry}
\usepackage[utf8]{inputenc} 
\usepackage[T1]{fontenc}    
\usepackage{hyperref}       
\usepackage{url}            
\usepackage{amsfonts}       
\usepackage{amsmath}
\usepackage{amsthm}
\usepackage{csquotes}
\usepackage{mhequ}
\usepackage[ruled,vlined]{algorithm2e}
\usepackage[comma,authoryear]{natbib}
\usepackage{color}
\usepackage{float}

\newtheorem{corollary}{Corollary}

\newtheorem{prop}{Proposition}
\newtheorem{remk}{Remark}

\newtheorem{proposition}{Proposition}

\begin{document}

\title{The limited-memory recursive variational Gaussian approximation (L-RVGA)  \\
}

\author{
  Marc Lambert  \\
  DGA/CATOD, Centre d'Analyse Technico-Op\'erationelle de D\'efense\\
  \&   INRIA - Ecole Normale Sup\'erieure - PSL Research university \\
  \texttt{marc.lambert@inria.fr} 
  \and
  Silvère Bonnabel \\
ISEA, Universit\'e de la Nouvelle-Cal\'edonie \\\& MINES ParisTech, PSL University, Center for robotics\\
  \texttt{silvere.bonnabel@mines-paristech.fr}  
  \and
  Francis Bach \\
  INRIA - Ecole Normale Sup\'erieure - PSL Research university  \\
  \texttt{francis.bach@inria.fr} 
}
 \date{}

\maketitle

\begin{abstract}
We consider the problem of computing a Gaussian approximation to the posterior distribution of a parameter given a large number $N$ of observations and a Gaussian prior, when the dimension of the parameter $d$ is also large. To address this problem we build on a recently introduced recursive algorithm for variational Gaussian approximation of the posterior, called recursive variational Gaussian approximation (RVGA), which is a single pass algorithm, free of parameter tuning. In this paper, we consider the case where the  parameter dimension $d$ is high, and we propose a novel version of RVGA that scales linearly in the dimension $d$  (as well as  in the number of observations $N$),  and which  only requires linear storage capacity in $d$. This is afforded by the use of a novel recursive expectation maximization (EM) algorithm applied for factor analysis introduced herein, to approximate at each step the covariance matrix of the Gaussian distribution conveying the uncertainty in the parameter. The approach is successfully illustrated on the problems of high dimensional least-squares and logistic regression, and generalized to a large class of nonlinear models. 
\end{abstract}

\section{Introduction}
\label{intro}	
In machine learning and statistics,  Bayesian inference aims at estimating the full posterior distribution of a parameter of interest $\theta$, to better track the uncertainty of the learning process. Exact Bayesian inference is generally not tractable, and approximations are required.   Variational approximation  \citep{Hinton93, VarEM} consists in  rewriting the estimation problem as an approximate optimization problem over  a restricted class of posterior distributions, such as Gaussian distributions, as advocated in the present paper.   

More precisely, assume we want to approximate the posterior distribution $p(\theta|Y_N)$ of a Bayesian parameter $\theta$ given $N$ observations $Y_N=y_1,\dots, y_N$. The variational Gaussian approximation \citep{Barber97,opper2009,Barber1} consists in minimizing the Kullback-Leibler divergence between this unknown posterior and a restricted class of parameterized distributions, namely a Gaussian  distribution $q(\theta|\mu,P) \sim \mathcal{N}(\theta|\mu,P)$, leading to the following optimization problem:
\begin{align}
\underset{\mu,P}{\min} &\quad KL\left(q(\theta|\mu,P) \| p(\theta | Y_N) \right)=\int q(\theta|\mu,P) \log \frac{q(\theta|\mu,P)}{p(\theta | Y_N)} d\theta. \label{batchKL} 
\end{align}

In the recursive variational approach, the observations at each time step consist only of the last sample~$y_t$, and the previous Gaussian estimate $q_{t-1}$ serves as a prior that sums up  information obtained so far. The posterior is then approximated by $p(\theta|y_t) \propto p(y_t|\theta)q_{t-1}(\theta)$. This leads to a recursive variational approximation problem that   writes at each step as follows:
\begin{align}
q_t(\theta) = \underset{ \mu,P} {\arg \min} \quad KL\left(\mathcal{N}(\theta|\mu,P) \| c_t p(y_t | \theta) q_{t-1}(\theta)\right),  \label{rvga}
\end{align}
where $c_t$ is a normalization constant that disappears when we consider the derivatives with respect to the variational parameters. It can be shown that the necessary stationary conditions for   \eqref{rvga} yield the following updates, which are implicit---see \citep{RVGA}:
\begin{align}
&\mu_t=\mu_{t-1}+P_{t-1}\mathbb{E}_{\theta \sim \mathcal{N}(\mu_t,P_t)}[\nabla_{\theta} \log p(y_t|\theta)]  \label{UpdateMean}\\ 
&P_t^{-1}=P_{t-1}^{-1} - \mathbb{E}_{\theta \sim \mathcal{N}(\mu_t,P_t)}[\nabla^2_{\theta} \log p(y_t|\theta)]. \label{UpdateCov}
\end{align}

We recognize a second-order online algorithm doing an adaptive gradient descent on the stochastic loss function $\ell_t(\theta)=-\log p(y_t|\theta,x_t)$. Although this algorithm seamlessly scales with the number of observations $N$, it has quadratic cost in the parameter dimension~$d$, hindering its use in high-dimensional problems. 
In this paper, we consider a factor analysis (FA) structure to approximate the precision matrix $P^{-1}$ at each time step with a \textquote{low rank + diagonal} matrix, that is,  $P^{-1} \approx WW^T+\Psi$ where $W \in \mathcal{M}(d \times p)$ is a rank~$p$ matrix with  $p \ll d$ and $\Psi  \in \mathcal{M}_d(\mathbb{R})$ is a diagonal matrix. This choice of decomposition makes it possible to manipulate matrices of size $d\times d$ while only storing $d+d\times p$ parameters in memory.

We reformulate the considered factorization in a recursive way such that the implicit equation \eqref{UpdateCov} becomes:
\begin{align}
&W_{t-1}W_{t-1}^T+\Psi_{t-1} - \mathbb{E}_{\theta \sim \mathcal{N}(\mu_t,(W_{t}W_{t}^T+\Psi_{t})^{-1})}[\nabla^2_{\theta} \log p(y_t|\theta)] \underset {\rm FA} \approx W_tW_t^T+\Psi_t.  \label{RecFA}
\end{align}
In the particular case of a linear regression model $y_t=x_t^T\theta + w_t$ with output noise $w_t  \sim \mathcal{N}(0, 1)$ and with $x_t$   the input associated to the observation $y_t$, this recursive update becomes explicit:
 \begin{align}
&W_{t-1}W_{t-1}^T+\Psi_{t-1} + x_tx_t^T \underset {\rm FA} \approx W_tW_t^T+\Psi_t.
\end{align}
We recognize a classical factor analysis problem for the scaled empirical covariance of the inputs $\sum_{t=1}^N x_tx_t^T$, but done recursively.

We show that the variational parameters $W_t$ and $\Psi_t$ can be computed through a recursive variant of the EM algorithm that is free of step tuning and provides a good approximation of the precision matrix after only one pass through the data. 

We then generalize this recursive EM approach in the nonlinear case. We have to assume additional approximations in this case: the Hessian term in \eqref{RecFA} is approximated with an explicit outer product and the expectations are approximated by sampling. To this aim we propose a novel sampling method to sample from $\mathcal{N}(\mu,(WW^T+\Psi)^{-1})$ without storing a $d \times d$ matrix, resulting in updates which scale  linearly with  $d$ in terms of both computation and storage costs. 
 
The paper is organized as follows: in Section \ref{rw} we discuss how our approach differs from several related studies on large-scale variational inference. In Section \ref{doubleVariational}, we model our problem as a two-stage  variational approximation scheme and show how the second stage can be recast as an expectation-maximization (EM) problem. Considering first the linear case, we reformulate our solution as a recursive EM algorithm with computation and storage cost being linear in $d$. In Section \ref{LRVGA}, we extend our approach to the nonlinear case using several approximations when updating the covariance matrix.  We also  introduce our novel method for sampling from a Gaussian whose dispersion is encoded by a precision matrix structured with a factor analysis model. 

In Section \ref{xp} our algorithm is evaluated  on synthetic data for large-scale matrix approximation and linear and logistic regression.
 
\section{Related work}
\label{rw}

Variational inference aims to approximate the distribution of latent parameters \citep{Hinton93, VarEM} in Bayesian machine learning, and has received a lot of attention in recent years for its ability to estimate a distribution without computing its normalization constant. Although the approximation may provide biased results, it converges much faster than the gold standard,  Markov Chain Monte Carlo or MCMC \citep{Andrieu2003}, which is eventually accurate but can be slow. Moreover, the model used in variational methods can be better specified using latent parameters \citep{Kohn16, Kohn13}.

To tackle large-scale datasets, stochastic solvers are used to compute the unknown parameters leading to general frameworks such as the black box variational inference algorithm \citep{ranganath2014}. In this article, we consider   variational Gaussian approximation (VGA) where the variational parameters boil down  to a mean vector and a covariance matrix. 
An approximation of the covariance matrix is needed to tackle high-dimensional problems in VGA. \citet{Barber1} have shown that, for a generalized linear model (GLM), the KL divergence is convex in the square root of the covariance matrix and have proposed different types of approximations for this form. Other works consider the factor analysis structure in the context of variational inference. \citet{Beal2000} used a mean field variational Bayes to obtain closed-form updates for the factor analysis parameters of each Gaussian of a mixture model. Their model is more general than ours since they use a mixture of Gaussian but all observations are required at each iteration. \citet{Ong2018} estimated the factor analysis parameters of a Gaussian using a stochastic gradient descent which allows the observation to be processed sequentially.
We also consider a factor analysis structure in this article but we do not directly optimize the parameters to avoid using stochastic gradient descent which can be difficult to tune.

\citet{SLANG} solved the variational problem in high dimension using a singular value decomposition (SVD) to compute a factor analysis structure for the precision matrix at each step of a natural gradient descent. The two approximation  schemes, Gaussian approximation, and factor analysis approximation are performed sequentially. Our contribution extends this approach and differs from it in several respects. First, we do not consider a natural gradient descent algorithm on the variational parameters but we use instead the RVGA scheme which requires no step tuning \citep{RVGA}. Second, we propose a parameter-free recursive EM algorithm to compute the factor analysis precision matrix approximation without explicitly using an SVD. SVD may not be efficient in high dimensions and the decoupled  approach proposed in \citet{SLANG} does not allow approximating accurately the correlation between the low-rank part and the diagonal part of the factor analysis structure. Moreover, we show that our approach has a nice interpretation as a two-stage variational inference that performs two \textquote{I-projections} \citep{Csiszar2004} at each step: the first one on the space of Gaussian distributions, the second one on the space of Gaussian distributions with a structured precision matrix constrained to be \textquote{diagonal + low-rank}.

Beyond variational inference, our work is also related to second-order stochastic optimization where the Hessian is approximated online to fit in the memory like in Adagrad \citep{Adagrad} or TONGA~\citet{Topmoumoute}. Interesting connections may be done also with large-scale Kalman filtering like the SEEK filter \citep{SEEK} which approximates the covariance matrix with a singular value decomposition or the ensemble filter \citep{Evensen94} which approximates it with sampling. 

\section{Two-stage variational approximation}
\label{doubleVariational}
 
To introduce a variational loss for factor analysis, we recast the original recursive problem  \eqref{rvga} as a two-stage variational problem as follows: 
\begin{align}
& \textbf{Two-stage variational approximation} \nonumber\\
 &q_t^*(\theta)= \mathcal{N}(\theta|\mu_t^*,P_t^*)=\underset{ \mu, P} {\arg \min} \quad KL( \mathcal{N}(\theta|\mu,P) \| c_t p(y_t|\theta)q_{t-1}(\theta)) \label{update1} \\
&q_t(\theta)= \mathcal{N}(\theta|\mu^*_t,\tilde{P}_t)=\underset{ W, \Psi} {\arg \min} \quad KL( \mathcal{N}(\theta|\mu^*_t,(WW^T+\Psi)^{-1}) \| q_t^*(\theta)). \label{update2}
  \end{align}

The first variational problem is the projection of the posterior distribution onto the space of Gaussian distributions whereas the second variational problem is the projection of a Gaussian onto the space of structured Gaussian having a low rank plus diagonal structure for the precision matrix. This factorization aims to exploit the parsimony of the covariance matrix to limit the memory requirement. 

The first projection was the object of our previous work    \citep{RVGA} and leads to the RVGA update giving the optimal solution for the mean \eqref{UpdateMean} and covariance \eqref{UpdateCov}, whereas the second one can be addressed through an expectation-maximization algorithm as we show now.

\subsection{Low rank + diagonal approximation via EM}\label{LREM}

We first show that the divergence in  \eqref{update2} may be algebraically related to a maximum likelihood problem. Indeed, since the means are the  same on both sides of the KL divergence, we may ignore them and rewrite the update as follows:
\begin{align}
&KL(\mathcal{N}(\theta|\mu_t^*,(WW^T+\Psi)^{-1})  \| \mathcal{N}(\theta|\mu_t^*,P_t^*)) \label{PbKL}\\
=\ &KL(\mathcal{N}(\theta|0,{P_t^*}^{-1}) \|  \mathcal{N}(\theta|0,WW^T+\Psi)) \label{leftKL} \\
=\ &\frac{1}{2} \mathrm{Tr} ((WW^T+\Psi)^{-1}{P_t^*}^{-1}) +\frac{1}{2} \log \det( WW^T+\Psi) -\frac{1}{2}  \log \det ({P^*}^{-1})-\frac{d}{2}. \label{leftKL2}
\end{align}

Let us consider $K$ samples $v_1,\dots, v_K$ supposed to be centered and such that their empirical covariance is $S_K=\frac{1}{K}\sum_{i=1}^Kv_iv_i^T={P_t^*}^{-1}$.  The log-likelihood on these samples is: 
\begin{align}
&\underset{W,\Psi}{\max}  \quad \log \mathcal{N}(v_1,\dots, v_K|0,WW^T+\Psi)\\
=\ &\underset{W,\Psi}{\max}  \quad - \frac{K}{2} \mathrm{Tr} ((WW^T+\Psi)^{-1} \frac{1}{K} \sum_{i=1}^Kv_iv_i^T) -\frac{K}{2}   \log \det (WW^T+\Psi)-\frac{K}{2}d\log(2\pi) \label{loglik1},
\end{align}
where we have used the relation $\sum_{i=1}^K v_i^T(WW^T+\Psi)^{-1}v_i=\mathrm{Tr}((WW^T+\Psi)^{-1} \sum_{i=1}^Kv_iv_i^T)$. We see that minimizing the divergence \eqref{leftKL2} is equivalent to maximizing the  log-likelihood \eqref{loglik1}.

The variational parameters can be obtained by zeroing the derivative of the maximum likelihood (ML) and using a singular value decomposition to find the solution. We will rather consider an expectation-maximization (EM) algorithm  \citep{Dempster77} which better scales to high-dimensional problems and which is guaranteed to increase the (total) likelihood at each step \citep[see][]{EMFA}. It can be shown that the EM approach computes implicitly a singular value decomposition to find the latent factors but in a memory-efficient way. The connection between the fixed-point ML, the fixed-point EM approach, and the associated eigenvalue algorithms is discussed in more detail in Appendix \ref{FixedPoint}. 

To apply the EM algorithm we need to introduce latent variables $z_1,\dots, z_K$ such that our samples take the form $v_i=Wz_i + \varepsilon_i$ where $\varepsilon_i \sim \mathcal{N}(0,\Psi)$ and $i=1,\dots, K$.

At the expectation step (E-step), the latent variables $z_i$ are estimated conditionally on the observation and the current variational parameters $W$ and $\Psi$ using the conditioning formula:
 \begin{align}
&p(z_i|v_i,W,\Psi)=\mathcal{N}(M^{-1}W^T\Psi^{-1}v_i,M^{-1}) \text{ where }   M=\mathbb{I}_p+W^T\Psi^{-1}W.
\end{align} 

At the maximization step (M-step), the variational parameters are adjusted to maximize the expected total likelihood defined by:
\begin{align}
&\mathbb{E}[\log p(v_1,z_1,\dots,v_K,z_K|W,\Psi)]=\mathbb{E}[\sum_{i=1}^K \log p(v_i|z_i,W,\Psi) + \sum_{i=1}^K \log p(z_i)] \\
&=\mathbb{E}[-\frac{K}{2} (v_i-Wz_i)^T\Psi^{-1}(v_i-Wz_i)-\frac{K}{2}\log \det \Psi  -\frac{K}{2}  \log  (2 \pi) + \sum_{i=1}^K \log p(z_i)],  \label{MLE-EMbatch}
\end{align}
where the expectations are taken over the conditional distribution $z_i \sim p(z_i|v_i,W,\Psi)$, where parameters $W,\Psi$ are there fixed to their current value.
After some calculations recapped in Appendix \ref{AppendixFixedPointEM}, the EM algorithm  for the factor analysis problem yields the following updates:
\begin{align}
\intertext{ \textbf{E-Step:}}
& \mathbb{E}[z_i|v_i]=M^{-1}W^T\Psi^{-1}v_i \nonumber \\
& \mathbb{E}[z_iz_i^T|v_i]=M^{-1}+\mathbb{E}[z_i|v_i]\mathbb{E}[z_i|v_i]^T.  \label{E-step}
\intertext{ \textbf{M-Step:}}
&W^{(n)}= \sum_{i=1}^K v_i \mathbb{E}[z_i|v_i]^T ( \sum_{i=1}^K \mathbb{E}[z_iz_i^T|v_i])^{-1} \nonumber\\
&\Psi^{(n)}={\rm{diag}} \Big(\frac{1}{K}\sum_{i=1}^K v_iv_i^T-W^{(n)}\frac{1}{K} \sum_{t=1}^K \mathbb{E}[z_i|v_i]v_i^T\Big), \label{M-step}
\end{align}
where the $^{(n)}$ stands for ``new''.

The expectation and maximization steps can be computed through a single update leading to the following fixed-point scheme (see Appendix \ref{AppendixFixedPointEM} for further details):
\begin{align}
W^{(n)}&=S_K\Psi^{-1}W (\mathbb{I}_p+M^{-1}W^T\Psi^{-1}S_K\Psi^{-1}W)^{-1},  \\
&\quad \text{where } M=\mathbb{I}_p+W^T\Psi^{-1}W  \nonumber \\
\Psi^{(n)}&={\rm{diag}} (S_K-W^{(n)}M^{-1}W^T\Psi^{-1} S_K) \label{FP-EM}\\
W&=W^{(n)},\quad \Psi=\Psi^{(n)}. \nonumber
\end{align}
This update no longer  depends on the samples $v_i$  nor on the latent parameters $z_i$ explicitly. It only makes use of $S_K=\frac{1}{K}\sum_{i=1}^Kv_iv_i^T={P_t^*}^{-1}$. The matrix ${P_t^*}^{-1}$ is the output of the R-VGA update  \eqref{UpdateCov} and depends on the previous estimate ${P_{t-1}}^{-1}$, supposed already in a factor analysis form, such that we have the relation:
\begin{align}
&{P_t^*}^{-1}=W_{t-1}W_{t-1}^T+\Psi_{t-1} - \mathbb{E}_{\theta \sim \mathcal{N}(\mu_t,P_t)}[\nabla^2_{\theta} \log p(y_t|\theta)].
\end{align}

Substituting this expression into our fixed point EM algorithm, we obtain a closed form for the second-stage variational update. But this update involves the $d \times d$ matrix $\nabla^2_{\theta} \log p(y_t|\theta)$, which does not hold in memory in large-scale problems. Moreover, this matrix is not well defined since it depends on an expectation under unknown parameters $\mathbb{E}_{\theta \sim \mathcal{N}(\mu_t,P_t)}$. 

We will first consider a linear model where the Hessian matrix and expectations disappear and updates reduce to simple recursive equations leading to Algorithm \ref{algorithmRecEM} which is memory-efficient.  We will then show in Section \ref{LRVGA} that this same algorithm can be used in the nonlinear case if we approximate the Hessian matrix and the expectations in a suitable way.

\subsection{The linear case: a recursive EM algorithm for factor analysis}\label{FA}
Let us consider for now the simpler case of a linear model $y_t=x_t^T\theta + w_t$, where $w_t  \sim \mathcal{N}(0, \sigma_w^2)$. The variance $\sigma_w^2$ is a scaling factor we will set equal to $1$. The RVGA updates \eqref{UpdateMean}-\eqref{UpdateCov}  is then equivalent to the following explicit updates, see \citealt{RVGA}:
\begin{align}
&\mu_t=\mu_{t-1}+P_{t}x_t(y_t-x_t^T\mu_{t-1})\label{mu2}\\
&P_t^{-1}=P_{t-1}^{-1} +x_tx_t^T. \label{CovInf}
\end{align}
 If we apply our two-stage variational approximation in this linear case, the first stage approximation \eqref{update1}  is exactly solved by the update above and the second stage approximation \eqref{update2} addresses limited-memory requirements. 
Using the factor analysis approximation in a recursive way, we obtain: 
\begin{align}
&\mu_t=\mu_{t-1}+(W_{t}W_{t}^T+\Psi_{t})^{-1}x_t(y_t-x_t^T\mu_{t-1}),\label{mu3} \\
&W_{t}W_{t}^T+\Psi_{t}  \underset{FA}{\approx} W_{t-1}W_{t-1}^T+\Psi_{t-1} + x_tx_t^T .  \label{LinearFA}
\end{align}
Using the Woodbury formula $(W_{t}W_{t}^T+\Psi_{t})^{-1}x_t=\Psi_t^{-1} (x_t- W_tM_t^{-1}(W_t^T \Psi_t^{-1}x_t))$ with $M_t=\mathbb{I}_p+W_t^T\Psi_t^{-1}W_t$,   \eqref{mu3}    can be rewritten in a limited-memory fashion involving only operations  linear in $d$:
\begin{align}
\eqref{mu3}\Leftrightarrow\mu_t=\mu_{t-1}+\Psi_t^{-1} (x_t- W_t(\mathbb{I}_p+W_t^T\Psi_t^{-1}W_t)^{-1}(W_t^T \Psi_t^{-1}x_t))(y_t-x_t^T\mu_{t-1})\label{mu22}.
\end{align} 

To address  \eqref{LinearFA},  we can use our fixed-point equation replacing the matrix $S_K$ in  \eqref{FP-EM}  by $W_{t-1}W_{t-1}^T+\Psi_{t-1}+x_tx_t^T$. This defines a recursive EM algorithm, namely Algorithm \ref{algorithmRecEM}, which consists in successively performing a few cycles on  the fixed-point equation \eqref{FP-EM}, where we expand and rearrange the terms to avoid any operations  requiring storing or multiplying $d\times d$ matrices. The Algorithm \ref{algorithmRecEM} is given in a more general setting where we solve a factor analysis approximation of the form:
\begin{align}
W_tW_t^T+\Psi_t \underset{FA}{\approx} \alpha_t (W_{t-1}W_{t-1}^T+\Psi_{t-1})+\beta_t x_tx_t^T, \label{NotNormalizedRecFA2}
\end{align} 
with $\alpha_t,\beta_t$ scalar constants being equal to 1 in this section. We will show in the next Section that we can use this same algorithm in the nonlinear case (or for the factorization of a covariance matrix) by using different values for the parameters $\alpha_t,\beta_t$ and replacing the entries $x_t$ by rectangular matrixes. 
 
Using the recursive EM in this way gives a factor analysis decomposition of the precision matrix ${P_t^*}^{-1}$, which corresponds to an increasing amount of information $P_0^{-1}+\sum_{i=1}^t x_ix_i^T$ as more inputs are processed, where $P_0$ corresponds to the covariance of the prior. 

This prior matrix $P_0$ can not be stored in memory and the algorithm \ref{algorithmRecEM} contains an initialization procedure which computes $W_0$ and $\Psi_0$ such that $W_{0}W_{0}^T+\Psi_{0} \approx P_0^{-1}$, see Appendix \ref{initPrior} for more details.

\begin{algorithm}[!ht]
\label{algorithmRecEM}
\SetAlgoLined
\KwResult{$W,\Psi$ }
Given $N$ inputs $x_1,\dots,x_N$ in high dimension $d$\; 
Given a latent dimension $p\ll d$\;
Given a prior covariance $P_0$\;

\textbf{-Initialization (consistent with $P_0$, see  \ref{initPrior})-}\\
Initialize $W \in \mathcal{M}(d \times p)$  and $\Psi  \in \mathcal{M}_d(\mathbb{R})$ diagonal:\\
$\Psi=\Psi_0$ \;
$W=W_0$\;
 \textbf{-Update-}\\
 \For{$t\leftarrow 1$ \KwTo $N$}{ 
  Access an input $x_t$\;
  \For{$k \leftarrow 1$ \KwTo $nbInnerLoop$}{ 
 $M=\mathbb{I}_p+W^T\Psi^{-1}W$\;
 $V= \beta_t x_t(x_t^T \Psi^{-1}W) + \alpha_t (W_{t-1}(W_{t-1}^T\Psi^{-1}W)+\Psi_{t-1}\Psi^{-1}W)$ \;
 $W^{(n)}=V(\mathbb{I}_p+M^{-1}W^T\Psi^{-1}V)^{-1}$\;
 $\Psi^{(n)}= \beta_t x_t*x_t+ \alpha_t  (W_{t-1}*W_{t-1} + \Psi_{t-1}) - W^{(n)}M^{-1}* V$ \;
 $W=W^{(n)}$ \;
 $\Psi=\Psi^{(n)}$ \;
 }
 $W_{t-1}=W$\;
 $\Psi_{t-1}=\Psi$\;
 }
 \caption{Recursive expectation-maximization algorithm for solving  factor analysis approximation \eqref{NotNormalizedRecFA2}} 
\tcc{From $1$ to $3$ inner loops ($nbInnerLoop$) may be sufficient to make the algorithm converge. The initialization derived in Section \ref{initPrior} compute $W_0$ and $\Psi_0$ such that $W_{0}W_{0}^T+\Psi_{0} \approx P_0^{-1}$. 
The operation \textquote{$*$}  which appears at the last line aims to compute  $X*Y={\rm{diag}}(XY^T)$  in a memory-efficient way. It is applied on two matrices $X$ and $Y$ of the same size $d \times p$. If $p=1$, this operator matches the component-wise operator $\odot$ : $X*Y=x \odot y$. If $p>1$, it writes $X*Y = \sum_{i=1}^p  x[.,i] \odot y [.,i]$.  
The diagonal matrix $\Psi$ is stored and used as a vector:   any multiplication of the diagonal matrix with a vector may be computed faster as an element-wise product: $\Psi^{-1}W=W/\psi$  where  $\psi=\rm{diag}(\Psi)$ and $/\psi$ operates on columns of $W$ and $W^T\Psi^{-1}$ gives as well  $W^T/\psi^T$ where $/\psi^T$ operates on the rows of $W^T$. }

\end{algorithm}

\section{The limited memory recursive variational Gaussian approximation (L-RVGA)}
\label{LRVGA}
In this Section, we consider the extension of our algorithm when the observations depend nonlinearly on the latent parameter $\theta$. We first extend the previous linear  case to  generalized linear models in Section \ref{LRVGA glm}. Section \ref{LRVGA general} considers the wholly general case, which necessitates additional approximations.
\subsection{L-RVGA  for generalized linear models}
\label{LRVGA glm}
The latter approach may be extended whenever the observation $y_t$ follows an exponential family distribution, which may be advantageous in the context of classification problems. An exponential family distribution \citep{Pitcher1979} takes the form: 
\begin{align}
&p(y_t)=c(y_t)\exp({\eta}^Ty_t-F( {\eta})),\label{exp:mo}
\end{align} 
where $F$ is the log partition function and $c(y)$ is a normalization function. It can represent a large family of distributions like the Gaussian, multinomial, or Dirichlet distribution. The natural parameter $\eta$ is related to the expectation parameter $m=\mathbb{E}[y]$ through the link function $g$ such that $\eta=g(m)$. In the machine learning framework, the natural parameter is also related to the input $x$ and the hidden latent parameter $\theta$ through a linear function $\eta=\theta^Tx$, this model is called the generalized linear model with a canonical natural parameter. The advantage of this model is that the Hessian term takes the form of a (generalized) outer product $-\nabla^2_{\theta} \log p(y_t|\theta)=x_t \nabla_\eta^2 F(\eta(\theta)) x_t^T$. In this case the RVGA updates  \eqref{UpdateMean}-\eqref{UpdateCov} become:
\begin{align}
&P_t^{-1}=P_{t-1}^{-1} + \gamma_t x_tx_t^T \label{rvgaGLM1} \\
&\mu_t=\mu_{t-1}+\xi_tP_{t-1}x_t ,\label{rvgaGLM2}
\end{align}
where we have introduced the scalar weighting parameter $\gamma_t$ and the scalar error $\xi_t$ which depend on $\theta$ as follows
\begin{align}
&\gamma_t=\mathbb{E}_{\theta \sim \mathcal{N}(\mu_{t},P_{t})}[\mathrm{Cov}(y_t|\theta)], \\
&\xi_t=y_t- \mathbb{E}_{\theta \sim \mathcal{N}(\mu_{t},P_{t})}[m(y|\theta)], 
\end{align}
and where we have used the properties of exponential family $\nabla_\eta F(\eta(\theta))=m(y|\theta)$ and $\nabla_\eta^2 F(\eta(\theta))=\mathrm{Cov}(y|\theta)$,  see  \cite{RVGA}. The factorized approximation of this update  can be performed through Algorithm \ref{algorithmRecEM}, since the factor analysis approximation $W_{t-1}W_{t-1}^T+\Psi_{t-1}  + \gamma_t x_tx_t^T  \underset{\rm FA}{\approx} W_{t}W_{t}^T+\Psi_{t}$ of \eqref{rvgaGLM1} fits into the problem \eqref{NotNormalizedRecFA2}, letting $\alpha_t=1,\beta_t=\gamma_t$. 
The computation of the scalar terms $\gamma_t$ and $\xi_t$ is nontrivial and requires resorting to approximations, though. Those approximations may rely on the tools developed in the next subsection, devoted to the general case. However, in the case of logistic regression with $m(y|\theta)=\frac{1}{1+\exp(-\theta^Tx)}$ and $\mathrm{Cov}(y|\theta)=m(y|\theta)(1-m(y|\theta))$ these parameters were shown in our prior work on RVGA \cite{RVGA} to be easily obtainable numerically and will serve as a baseline for our experimentation in Section~\ref{LogReg}.

\subsection{Limited memory RVGA (L-RVGA): nonlinear model}
\label{LRVGA general}
As before, we suppose that the observation $y_t$ follows the exponential family distribution \eqref{exp:mo} but with a nonlinear dependence of the form $\eta=h(\theta,x)$ with $h$ is a nonlinear function. In our setting, we need to approximate the matrix $P_t^{-1}$ recursively, as follows
 \begin{align}
P_t^{-1}&= W_{t-1}W_{t-1}^T+\Psi_{t-1} - \mathbb{E}_{\theta \sim \mathcal{N}(\mu_{t},P_{t})}[\nabla^2_{\theta} \log p(y_t|\theta,x_t)] \label{upp;eq2} \\
&\approx  W_{t-1}W_{t-1}^T+\Psi_{t-1} + X_tX_t^T \label{upp;eq3} \\
&\underset{\rm FA}{\approx} W_{t}W_{t}^T+\Psi_{t}, \nonumber
\end{align} 
where $X_t$ is a rectangular matrix  supposed to fit into memory which will be defined more precisely  in Section \ref{OuterApprox}. Using this approximation, we can run the recursive EM Algorithm \ref{algorithmRecEM}  in a memory-efficient way, replacing the input $x_t$ by the matrix $X_t$.  The mean is computed as before:
 \begin{align}
 &\mu_t=\mu_{t-1}-(W_{t-1}W_{t-1}^T-\Psi_{t-1})^{-1}\mathbb{E}_{\theta \sim \mathcal{N}(\mu_{t},P_{t})}[\nabla_{\theta} \log p(y_t|\theta,x_t)], \label{upp;eq} 
\end{align}
where we can use  the Woodbury formula  again to compute efficiently the weighting term $(W_{t-1}W_{t-1}^T-\Psi_{t-1})^{-1}$.

A practical implementation of these updates raises the following difficulties:
\begin{enumerate}
\item Computing $P_t^{-1}$ and $\mu_t$ in \eqref{upp;eq2}-\eqref{upp;eq} involves computing an expectation over a distribution parameterized by $\mu_t$ and $P_t$, i.e., the scheme is implicit. 
\item  As no closed form is generally available for the computation of the expectations $\mathbb{E}_{\theta \sim \mathcal{N}(\mu,P)}$, one may resort to Monte-Carlo sampling. However, in this paper we maintain a limited memory approximation of the \emph{inverse} of the covariance matrix, that is, $P^{-1}=WW^T+\Psi$ and we need to sample from a Gaussian $\mathcal{N}(\mu,P)$ without storing  or inverting  a $d\times d$ matrix. 
\item The Hessian matrix $\nabla^2_{\theta} \log p(y_t|\theta,x_t)$ (or its averaged value) must be computed and stored with linear cost in the dimension $d$. 
 
\end{enumerate} 
In the remainder of the section, we address all these points. 

\subsubsection{Using extra-gradients for the implicit scheme}\label{extrasec}
The more direct way to address Point $1$ above is to open the loop and to replace the expectations  under  $\mathcal{N}(\mu_{t},P_{t})$ with expectations under $\mathcal{N}(\mu_{t-1},P_{t-1})$. However, experiments we conducted showed that this naive scheme can lead to instability. The importance of managing the implicit scheme was the object of our previous work  \citep{RVGA}, where we developed closed-form formulas to solve the implicit scheme in the linear and logistic regression case. In the general case, one may resort to extra-gradient, i.e., to update first the covariance and then the mean, and to iterate twice: 
\begin{equation}
\label{mirrorProx} 
 \begin{aligned}
 &\textbf{Iterated RVGA}\\
&\color{red} \mathbf{\hat{P}_{t}} \color{black}^{-1}=P_{t-1}^{-1} - \mathbb{E}_{\theta \sim \mathcal{N}(\mu_{t-1},P_{t-1})}[\nabla^2_{\theta} \log p(y_t|\theta)]\\
&\color{red}  \mathbf{\hat{\mu}_{t}} \color{black}=\mu_{t-1}+\color{red} \mathbf{\hat{P}_{t}} \color{black}\mathbb{E}_{\theta \sim \mathcal{N}(\mu_{t-1},P_{t-1})}[\nabla_{\theta} \log p(y_t|\theta)]\\
&\color{blue} \mathbf{P_{t}} \color{black}^{-1}=P_{t-1}^{-1} - \mathbb{E}_{\theta \sim \mathcal{N}(\color{red}  \mathbf{\hat{\mu}_t} \color{black},\color{red} \mathbf{\hat{P}_{t}} \color{black})}[\nabla^2_{\theta} \log p(y_t|\theta)]  \\
&\mu_{t}=\mu_{t-1}=\color{blue} \mathbf{P_{t}} \color{black} \mathbb{E}_{\theta \sim \mathcal{N}(\color{red}  \mathbf{\hat{\mu}_t} \color{black},\color{red} \mathbf{\hat{P}_{t}} \color{black})}[\nabla_{\theta} \log p(y_t|\theta)].
\end{aligned}
\end{equation}

It turns out that this iterated scheme is equivalent to the ``Mirror Prox'' algorithm \citep{MirrorProx}, a.k.a., extra-gradient, with a unit step size   and applied to the function $f(\mu,P)=\mathbb{E}_{\theta \sim \mathcal{N}(\mu,P)}[\log p(y|\theta)]$. Mirror Prox is known to help convergence on a large set of problems (convex optimization, variational inequalities) and we have experimentally observed that this iterated scheme reduces significantly the bias of our estimator (see Appendix \ref{Extragrad} for further details). 

However, when we combine extra-gradient with factor analysis,  the extra covariance update can make the Mirror Prox scheme unstable. We have observed it is then  preferable   to skip the extra covariance update, that is, the third line of the iterated scheme above (see Appendix \ref{Extragrad}). This choice has been made in Section \ref{NonLin} dedicated to experiments in the general case. 

\subsubsection{Gaussian sampling from the precision matrix}\label{samplesec}
To address Point $2$ above, we need to sample efficiently from $\mathbb{E}_{\theta \sim \mathcal{N}(\mu,(WW^T+\Psi)^{-1})}$ without storing a $d \times d$ matrix. This problem was already addressed by \citet{SLANG} and \citet{Sampling2014} using a square-root form and two Cholesky decompositions on the latent space. However, we propose here a faster method inspired by the ensemble Kalman filter \citep{Evensen94} which does not require computing any Cholesky decomposition. Our method is close to the one developed in \cite{Orieux12} for a sparse precision matrix but is more suitable for factor analysis.

\begin{prop}
Let us define the quantities:
  $M=\mathbb{I}_p+W^T\Psi^{-1}W$ 
and  $L=\Psi^{-1}WM^{-1}$. Draw $x\sim\mathcal N(0,\Psi^{-1})$ and   $\epsilon \sim\mathcal N(0,\mathbb{I}_p)$ independently, and define 
$$x^+=x+L(\epsilon-LW^Tx)=(\mathbb{I}_d-LW^T)x+L\epsilon.$$
We have then $x^+\sim\mathcal N(0,P)$, with $$
P=(WW^T + \Psi)^{-1} =\Psi^{-1}-\Psi^{-1}W(\mathbb{I}_p+W^T\Psi^{-1}W)^{-1}W^T\Psi^{-1}.$$ 
\end{prop}
\begin{proof}     
We have obviously $E(x^+)=0$. Moreover using the independence of the variables
\begin{align}
 E\left(x^+(x^+)^T\right)&=(\mathbb{I}_d-LW^T)\Psi^{-1}(\mathbb{I}_d-LW^T)^T+L\mathbb{I}_pL^T\\&=
 \Psi^{-1}-LW^T\Psi^{-1}-\Psi^{-1}WL^T+L[W^T\Psi^{-1}W+\mathbb{I}_p]L^T\label{secondline}\\
 &=\Psi^{-1}-\Psi^{-1}W(\mathbb{I}_p+W^T\Psi^{-1}W)^{-1}W^T\Psi^{-1},
\end{align}
since the three rightmost terms of \eqref{secondline} are all equal to $\pm\Psi^{-1}W(\mathbb{I}_p+W^T\Psi^{-1}W)^{-1}W^T\Psi^{-1}
$. 
\end{proof}

This suggests that starting from a decomposition of the form $WW^T + \Psi$ of the precision matrix $P^{-1}$, we know how to sample from a law $\mathcal N(0,P)$. We need to draw $x_1,\dots,x_K$ from $\mathcal N(0,\Psi^{-1})$  and $\epsilon_1,\dots,\epsilon_K$ from $\mathcal N(0,\mathbb{I}_p)$, and to let $x_i^+=x_i+L(\epsilon_i-LW^Tx_i)=(\mathbb{I}_d-LW^T)x_i+L\epsilon_i$, for $1\leq i\leq K$. All operations involved are linear in $d$, so that drawing $K$ samples is of order $O(Kd)$. As soon as $K$ is kept moderate with respect to $d$, the complexity is linear in $d$. 

\subsubsection{Dealing with the $d\times d$ Hessian matrix}\label{ggnsec}
In this paragraph, we address Point $3$ of the list above, where we propose an approximation of the Hessian based on an outer product.   
The Gauss-Newton approximation, which consists in approximating the Hessian with the outer product of the gradients, called also empirical Fisher, does not capture well second-order information (see for instance \citealt{kunstner2019}). As we have supposed the probability distribution belongs to an exponential family, we should rather consider the Generalized Gauss-Newton (GGN) approximation \citep{martens2014} which better approximates the local curvature \citep{kunstner2019}.  The GGN approximation exploits the structure of the exponential family  $p(y_t|\eta=h(\theta,x))$ as a composition of functions involving the natural parameter $\eta$ and the nonlinear model $h$. If the Hessian of $- \log p$ with respect to $\theta$ can be complicated, the Hessian with respect to $\eta$ is simply $\mathrm{Cov}(y|\theta)$.  We can use this property through the chain rule  to generalize the Gauss-Newton approximation. Our expected Hessian term can be written as follows:
\begin{align} 
 -\mathbb{E}_{\theta}   [\nabla^2_{\theta}  \log p(y_t|\theta)] \approx \mathbb{E}_{\theta}[\frac{\partial h}{\partial \theta}  \mathrm{Cov}(y|\theta)  \frac{\partial h}{\partial \theta}^T]= \mathbb{E}_{\theta} [\mathbb{F}(\theta)],
\end {align} 
where $\mathbb{F}(\theta)=\mathbb{E}_{y \sim p(y|\theta)} [- \nabla^2_{\theta}  \log p(y|\theta)]$ is the Fisher matrix. Under the GGN approximation, the covariance update no longer depends on the labels $y_t$ and only depends  on the inputs  $x_t$,  as in the linear case. The derivation is detailed in Appendix \ref{ProofLemmaGGN}. We can now combine this approximation with the approximation of the expectation and the implicit scheme to implement our final update. 

\subsubsection{Final algorithm}\label{OuterApprox}
We now define the form of matrix $X_t$ involved in the outer product in \eqref{upp;eq3}. Assuming that we make the scheme explicit using extra-gradients  (as described  in  Section \ref{extrasec}) and that we generate $K$ samples to approximate the expectation with Ensemble sampling described in  Section \ref{samplesec}, the GGN approximation becomes:
\begin{align}
P_t^{-1}  & \underset{\rm \hspace{0.3cm} GGN \hspace{0.3cm} }{\approx} W_{t-1}W_{t-1}^T+\Psi_{t-1} + \mathbb{E}_{\theta \sim \mathcal{N}(\mu_t,P_t)}[\frac{\partial h}{\partial \theta}  \mathrm{Cov}(y|\theta)  \frac{\partial h}{\partial \theta}^T] \\
& \underset{\rm Extragrad}{\approx} W_{t-1}W_{t-1}^T+\Psi_{t-1} + \mathbb{E}_{\theta \sim \mathcal{N}(\mu_{t-1},(W_{t-1}W_{t-1}^T+\Psi_{t-1})^{-1})}[\frac{\partial h}{\partial \theta}  \mathrm{Cov}(y|\theta)  \frac{\partial h}{\partial \theta}^T] \\
&\underset{\rm Sampling}{\approx} W_{t-1}W_{t-1}^T + \Psi_{t-1} + \frac{1}{K}\sum_{i=1}^K \frac{\partial h}{\partial \theta}(\theta_i) \mathrm{Cov}(y|\theta_i)^{1/2}  \mathrm{Cov}(y|\theta_i)^{1/2}  \frac{\partial h}{\partial \theta}^T (\theta_i)\\
&\quad = W_{t-1}W_{t-1}^T + \Psi_{t-1} + \frac{1}{K}\sum_{i=1}^K c_i c_i^T\quad \text{ where } \quad c_i=\frac{\partial h}{\partial \theta}(\theta_i)  \mathrm{Cov}(y|\theta_i)^{1/2}\\
&\quad = W_{t-1}W_{t-1}^T + \Psi_{t-1} + X_tX_t^T \quad \text{ where } \quad X_t=\frac{1}{\sqrt{K}} \begin{pmatrix} c_1 \cdots c_K \end{pmatrix} \\
&\underset{\rm \hspace{0.3cm}  FA \hspace{0.3cm} }\approx W_{t}W_{t}^T + \Psi_{t},  \label{outerGrad}
\end{align} 
where the last line refers to the FA approximation developed in Section \ref{LREM} applied to the mini-batch matrix $X_t$ of size $d \times K$. As long as the number of samples $K \ll d$, the  memory cost is kept linear in $d$. Regarding the other approximations, the approximation of the expectation with sampling seems not very sensitive and we recommend using a limited number of samples to reduce the memory cost. Indeed, we have observed few samples are sufficient to obtain a good approximation of the logistic regression problem considered in Section \ref{NonLin}.

\section{Experiments}
\label{xp}
Experiments have been performed  on synthetic data for different classes of problems, from simpler linear problems to more difficult nonlinear problems. The first class of problems addressed in Section \ref{largeScaleCov} deals with the approximation of large-scale covariance matrices that are too large to hold in memory. We show our approach requires far less memory than the batch EM algorithm which requires storing a $d \times d$ matrix in memory and competes with the online EM algorithm \citep{onlineEM}. These results are directly applied to linear regression in Section \ref{LinReg} where we propose computationally cheap updates for the linear Kalman filter in high dimensional spaces. In Section \ref{LogReg}, we consider the logistic regression problem for which we can solve the recursive variational scheme in a closed form at each step \citep{RVGA}. Finally, we turn to more general cases in  Section \ref{NonLin} where we assess our approximation method  for L-RVGA updates in the general case. 
 The logistic regression problem, where we know how to compute the expectation analytically, offers a baseline for comparison purposes. 
  
 We observe the mirror-prox method combined with ensemble sampling reaches the  baseline results in terms of KL divergence. This is promising regarding the generalization of our algorithms to arbitrary nonlinear problems. 
 
\subsection{Limited-memory approximation of large-scale covariance matrices }

In this section, we assess the limited-memory recursive EM algorithm on large-scale empirical covariance matrices. In this setting, we do not update a precision matrix but a covariance matrix with a moving average, this amounts to using  Algorithm \ref{algorithmRecEM} with $\alpha_t=(t-1)/t$ and $\beta_t=1/t$ (see Appendix \ref{FAforMatrix} for details).  To assess the method, we generate the data from an actual low-rank covariance plus diagonal matrix $S=\mathrm{Diag}(\psi)+WW^T$  for which the factor analysis approximation can be exact. We construct such a matrix with random parameters where $W$ is of size $d \times p$ and $\psi$ is a vector of size $d$. We then generate $N$ samples from the zero-mean Gaussian distribution equipped with this covariance matrix $S$. 
To make the experiments more appealing, we have also assessed our algorithm on the NIPS Madelon dataset and the Breast Cancer real dataset.  We have kept only the input data, using the LIBSVM library \citep{libsvm}, and normalized them using the same process as for the synthetic dataset.

 %The first considered experimental problem is large-scale covariance approximation. It is first shown to be amenable to the class of problems we have considered hitherto. Before moving on to the actual numerical experiments, we discuss the initialization of  Algorithm \ref{algorithmRecEM}. 

%\subsubsection{Factorization of a synthetic large scale covariance matrix }
\label{largeScaleCov}

\begin{figure}[!ht]
\includegraphics[scale=1]{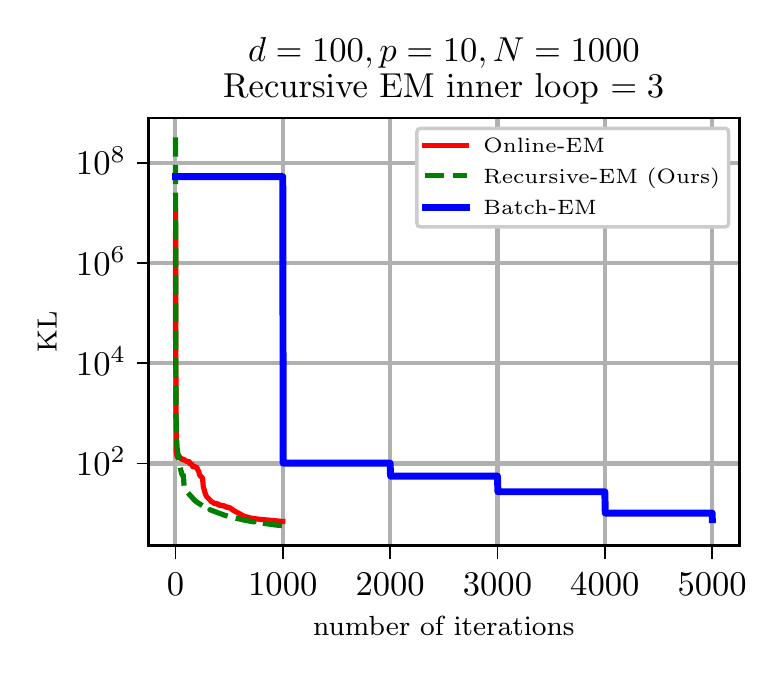}
\includegraphics[scale=1]{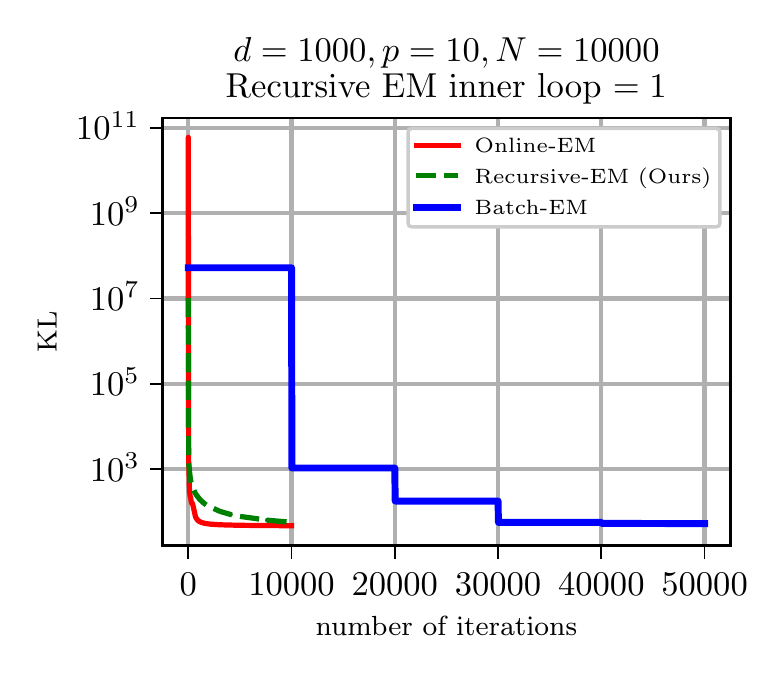}
\includegraphics[scale=1]{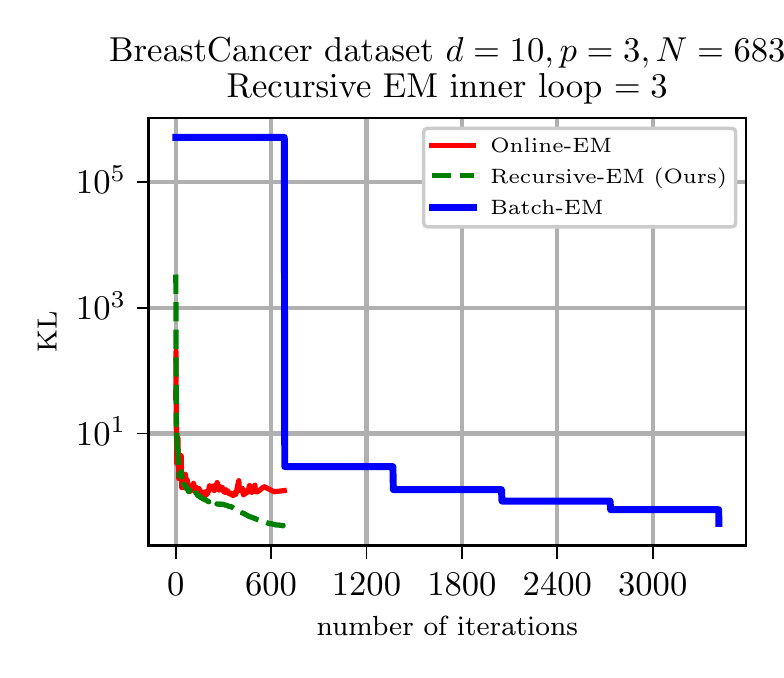}
\includegraphics[scale=1]{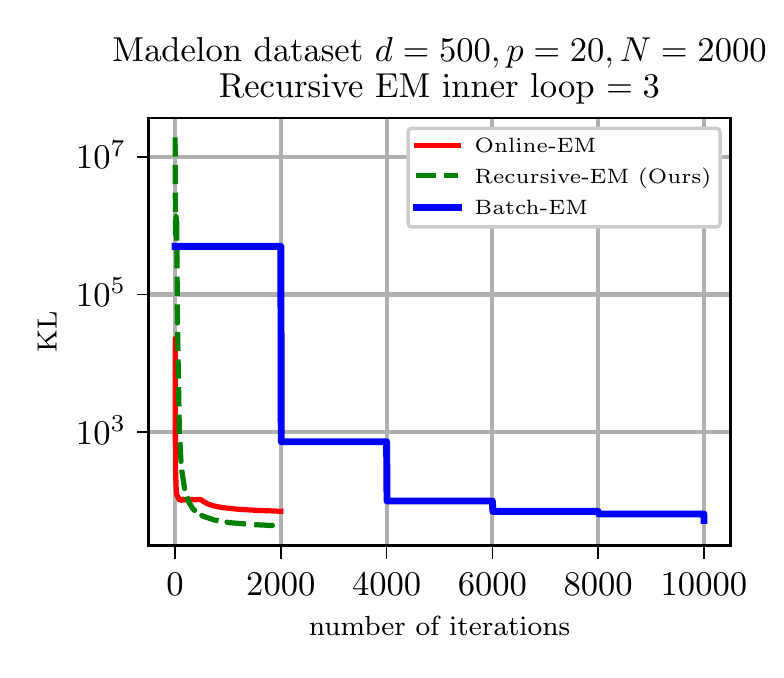}

\caption{Factorization of a synthetic large scale covariance matrix of dimension $d=100$ (upper left plot) and $d=1000$ (upper right plot). 
The lower plots show additional results on two  covariance matrices computed from the Madelon and Breast Cancer datasets. We show in dash green line our recursive EM algorithm, which is parameter-free, and in the red plain line, the online EM tuned with a learning rate $\gamma=\frac{1}{t^{0.6}}$ (see Appendix \ref{onlineEM} for implementation details).  
We show also the batch EM, in the solid blue line. The online and batch methods should be compared with respect to the number of observations processed. Batch estimation requires processing all the data several times, although more efficiently since 5 steps are enough to converge. The online variants use only a single passe.  With $3$ inner loops in dimension $d=100$ (upper left plot), the recursive EM is clearly superior to the online EM.  In the higher dimension (upper right plot), we achieve a similar loss to the online EM using only 1 inner loop in the recursive EM. Moreover, the recursive EM gives better results than the other algorithms on real datasets as illustrated in the lower left plot and lower right plot.} 
\label{figureXPCovHD}
\end{figure}
 
We run the batch EM algorithm on this set of $N$ samples and our limited memory EM algorithm recursively on each sample. The recursive EM generally converges to the result given by the batch version after only one pass through the data as shown in Figure \ref{figureXPCovHD}. Our recursive EM is also compared to the online EM algorithm \citep{onlineEM} which is also memory efficient. The derivation of the online EM for factor analysis and the implementation detailed have been moved to Appendix \ref{onlineEM}. 
While our method is parameter-free it yields results similar  to the online EM on the synthetic dataset and better results on the real datasets as shown in upper plots of Figure \ref{figureXPCovHD}. 

The advantage of online or recursive EM versions for factor analysis is that they maintain a memory cost linear in~$d$, which is an important feature in high dimension. To illustrate this feature, we performed large-scale covariance matrix approximation in dimension $1$ million. We see  in Table~\ref{tableXPCovHD} that the proposed algorithm may be performed using a regular laptop, contrary to the batch EM.  

\begin{table}[!ht]
\centering
\begin{tabular}{|c|c|c|c|c|}
  \hline
 Algorithm & Nb. iterations  &  Time per iteration & Memory cost \\
  \hline
    Batch EM (estimated)&  Fixed ($\ll N$) & -& \textbf{8000 GB}\\
      Online EM-$1$ inner loop & $N$ & 1 s & \textbf{1.6 GB}\\
  Recursive EM-$1$ inner loop & $N$ & 2 s  & \textbf{0.8 GB}\\
  \hline
\end{tabular}
\caption{\textbf{Memory test for large scale matrix factorization with $d=10^6$, and $p=100$}: Test dedicated for memory requirement, the recursive and online EM have been executed on a laptop (Intel core i7 at 2.3 GHz on CPU) in dimension one million, the memory cost for the batch EM have been estimated as it scales quadratically in $d$ and did not fit the memory, so no time is available for it. A reduced number of samples $N$ have been considered in this experiment since they do not influence the memory cost. Recursive EM requires more operations by iteration than online EM but consumes less memory than online EM and much less than batch EM.}
\label{tableXPCovHD}
\end{table}

\subsection{Application to Bayesian linear regression }
\label{LinReg}
We apply in this section the previous covariance approximation to derive a limited memory version of the linear regression problem as described in Section \ref{FA}. In the linear general  case,  we have an analytical form of the solution given by the linear Kalman filter with a static state which is our baseline. 

The linear Kalman filter is a parameter-free online algorithm that gives the exact solution of a linear regression problem after only one pass on the dataset, that is, it recursively computes exactly the posterior given all the data it has considered so far. However, it requires estimating online the  $d \times d$  covariance matrix of the inputs, which is intractable  in high dimension. This matrix is approximated with the recursive EM algorithm in this experiment to provide large-scale linear regression. 

The inputs $x_t$ are generated using  the following Gaussian distribution: 
\begin{align}
x_t \sim \mathcal{N}(0,C) \quad \text{with} \quad
C=
M^T
{\rm Diag}( 1, 1/2^c,\cdots, 1/d^c) 
 M,
\end{align}
where $M$ is an orthogonal rotation matrix and $c$ a coefficient driving the condition number ($c=1$ by default). Since the matrix $C$ is rotated, it is not directly available in a factor analysis form. The inputs are normalized on average. 

The outputs $y_t$ are generated with a normal noise:  
\begin{align}
&y_t=x_t^T\theta^* + \varepsilon_t \text{ with } \varepsilon_t \sim \mathcal{N}(0,1), 
\end{align}

We consider the Bayesian setting where $\theta$ is supposed to be a Gaussian random value and we want to estimate the parameters of its distribution $\mathcal{N}(\mu,P)$ where $P^{-1}=WW^T+\Psi$, given all the examples $(x_t,y_t)$ seen so far. We suppose that the prior follows a centered isotropic Gaussian distribution $\theta_0 \sim\mathcal{N}(0,\sigma_0^2 \mathbb{I}_d)$ and compute the initial values of the factor analysis parameters $W_0$ and $\Psi_{0}$ as detailed in  Appendix \ref{initPrior}. Our recursive scheme is sensitive to the parameter $\sigma_0$: a  high value of $\sigma_0$ (flat prior) may lead to bad conditioning at the first steps and make our recursive algorithm diverge. A low value for $\sigma_0$ (strong prior)  speeds up the convergence of the algorithm. This value defines the shape of the posterior and the difficulty of the estimation problem and can be used to test the robustness of our algorithm in the following experiments. In this Section, we consider a prior deviation $\sigma_0=1$. 

When $d=p$ we have checked that our algorithm matches the linear Kalman filter's estimates.  To assess how the factor analysis approximation degrades the results, we have done experiments for different values of the latent space dimension $p$. We see in Figure \ref{figureXPLin2} that the divergence decreases globally for all latent dimensions even if for lower values of $p$ the divergence may temporarily increase. As expected the convergence is faster for higher values of $p$. 

Table \ref{tableXPLin} details the memory cost for higher dimensions.

\begin{remk}[Uncertainty quantification and calibration] 

Using a lower value for $p$ may lead to a poor estimation of the uncertainty. In particular, the estimated covariance may not be  consistent with the empirical covariance. Variational inference is known to underestimate the target distribution but in critical applications, we may often prefer to overestimate the distribution. A common approach in Kalman filtering is to add a process noise to recover consistency.  In our setting, it will be equivalent to adding a  covariance matrix of noise $Q_t$ in the L-RVGA update \eqref{CovInf}. The covariance update will then rewrite $P_{t|t-1}=P_{t-1}+Q_t$ and $P_t^{-1}=P_{t|t-1}^{-1} +x_tx_t^T$. To formulate this update in terms of the factor analysis parameters $W_t$ and $\psi_t$, we can assume $Q_t=\eta_t P_{t-1}$ which corresponds to a fading memory filter. The parameter(s) of the covariance of noise may be learned online using a variational approach or an adaptive filter that adapts the covariance based on the prediction error (see for instance \cite{Zhang20}).  \end{remk}

\begin{figure}[!ht]
\begin{center}
\includegraphics[scale=1]{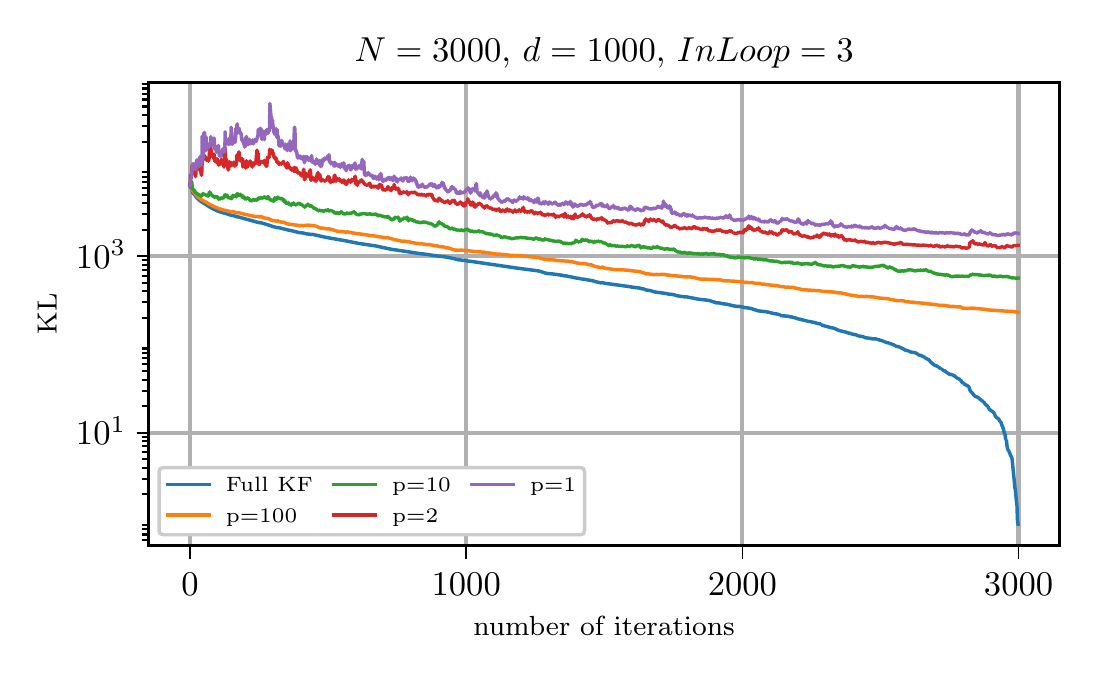}

\vspace*{-.2cm}

\caption{Sensitivity to the latent dimension $p$. We show the error as the KL divergence between the true posterior distribution and the L-RVGA posterior. The full Kalman error  converges exactly to $0$, as the Kalman filter exactly finds the posterior. 
%The approximated version with $p=100$ runs a close second. 
When $p$ is very low,   oscillations occur.}
\label{figureXPLin2}
\end{center}
\end{figure}

\begin{table}[!ht]
\centering
\begin{tabular}{|c|c|c|c|c|}
  \hline
 Algorithm &  Nb. iterations  &  Time per iteration & Memory cost \\
  &   &  (for $3$ inner loops)  &  \\
  \hline
Full Kalman & $N$  & -  & \textbf{8000 GB}\\
L-RVGA ($p=100$)  & $N$ & 6 s & \textbf{816 MB}\\
L-RVGA ($p=10$)   & $N$ & 0.7 s & \textbf{96 MB}\\
L-RVGA ($p=2$)   & $N$  & 0.2 s  & \textbf{32 MB}\\
L-RVGA ($p=1$)   & $N$ & 0.06 s & \textbf{24 MB}\\
  \hline
\end{tabular}
\caption{\textbf{Memory requirements test for linear regression with $d=10^6$ and variable $p$}: The L-RVGA has been executed on a laptop (Intel core i7 at 2.3 GHz on CPU) in dimension one million, the memory cost for the full Kalman has been estimated as it scales quadratically in $d$ and did not fit into the memory, so that  execution  time is not available for it. A reduced number of samples $N$ has been considered in this experiment since they do not influence the memory cost. When $p$ is low, the cost in memory as well as  the time per iteration (with $3$ inner loops for the recursive EM) is drastically reduced.}
\label{tableXPLin}
\end{table}

\subsection{Application to Logistic regression }
\label{LogReg}
We now apply the method to derive a limited-memory version of a logistic regression problem which is a particular case of the generalized linear model described in Section \ref{LRVGA glm}. In this model, we seek to learn the parameter $\theta$ encoding the  hyperplane from $N$ examples $(x_t,y_t)$.  We consider the Bayesian setting where $\theta \sim \mathcal{N}(\mu,P)$ and the observations are generated from $y_t=\sigma(\mu^Tx_t)$, where $\sigma$ denotes here the logistic function. The estimation of $\mu_t$ and $P_t$   are given by the RVGA updates for general linear model \eqref{rvgaGLM1}-\eqref{rvgaGLM2}. These updates involve  the expectation of the gradient and the Hessian of the logistic loss which can be approximated as follow: 
\begin{align}
\mathbb{E}_{\theta}[\nabla_\theta \ell_t(\theta)]&=-y_tx_t + \mathbb{E}_{\theta}[\sigma(x_t^T\theta)]x_t \approx -y_tx_t + \sigma(k_t x_t^T\mu_t)   \label{ExpectLog1}\\
\mathbb{E}_{\theta}[\nabla^2_\theta \ell_t(\theta)]&=\mathbb{E}_{\theta}[\sigma(x_t^T\theta)(1-\sigma(x_t^T\theta))]x_tx_t^T \approx k_t\sigma^\prime(k_t x_t^T\mu_t)x_t x_t^T  \label{ExpectLog2}\\
&\text{ where }   k_t= \frac{\beta}{\sqrt{x_t^TP_{t}x_t+\beta^2}} \text{ and }  \beta=\sqrt {\frac{8}{\pi}}.  \label{kt:def}
\end{align}  

These equations were derived in our previous work (see \citep{RVGA}, Section 4.1)  and come  from the approximation of the  logistic function $\sigma$ with the inverse probit function $\phi$ \citep{barber1998}:$\sigma(x) \approx \phi(\frac{1}{\beta}x)=\frac{1}{2}(1+erf(\frac{x}{\sqrt{2}\beta})). $   
We  can then rewrite \eqref{rvgaGLM1}-\eqref{rvgaGLM2} as:
\begin{align}
&\mu_t=\mu_{t-1} + P_{t-1}x_t (y_t-  \sigma(k_t x_t^T\mu_t))  \label{updateMu} \\
&P_t^{-1}=P_{t-1}^{-1}+k_t\sigma^\prime(k_t x_t^T\mu_t)x_t x_t^T. \label{updateP}
\end{align} 
The scheme is now  implicit owing only to the two scalar parameters $\nu=x_t^TP_{t}x_t$ and $\alpha=x_t^T\mu_t$ which can efficiently be computed by solving a scalar fixed point equation with a Newton solver, see \citealt{RVGA}, Section 4.2. 

We can in turn apply our recursive EM algorithm on the input $X_t:=x_t \sqrt{k_t\sigma^\prime(k_t x_t^T\mu_t)}$ for high-dimensional logistic regression. To assess the performance, we generate $N$ synthetic pairs $(x_t,y_t)$ where the inputs are generated from the same Gaussian distribution as in the linear regression case. The inputs are normalized in mean and the prior is set to $\sigma_0=4$ which corresponds to a sharp posterior distribution.  In Figure~\ref{figureXPLog1} we plot the KL divergence between our current Gaussian estimate $\mathcal{N}(\mu_t,P_t)$ and the true posterior, that is  $KL\left(\mathcal{N}(\mu_t,P_t)||p(\theta|y_1,x_1,\dots y_N,x_N)\right )$. In logistic regression, we have a simple expression for the posterior at each $\theta$, up to a normalizing constant. As the former left KL divergence is an expectation under $\mathcal{N}(\mu_t,P_t)$, it may be approximated via Monte-Carlo sampling. This offers a way to perform relative comparisons between the algorithms in terms of divergence with respect to the true (unnormalized) posterior. For more details see \citealt{RVGA}.

We compare the results  for different values of the parameter $p$ with the Laplace approximation which gives a batch approximation of $\mu$ and $P$. The L-RVGA converges to the batch Laplace approximation and may even yield lower divergence. This is because the covariance given by the Laplace approximation spills out the true posterior to regions of very low probability whereas the L-RVGA avoids them. Contrary to the linear case, even with $p=1$, we obtain very good results. The memory requirements  in dimension one million yield identical costs to those reported in Table \ref{tableXPLin}  for linear regression. 
  
\begin{figure}[!ht]
\includegraphics[scale=1]{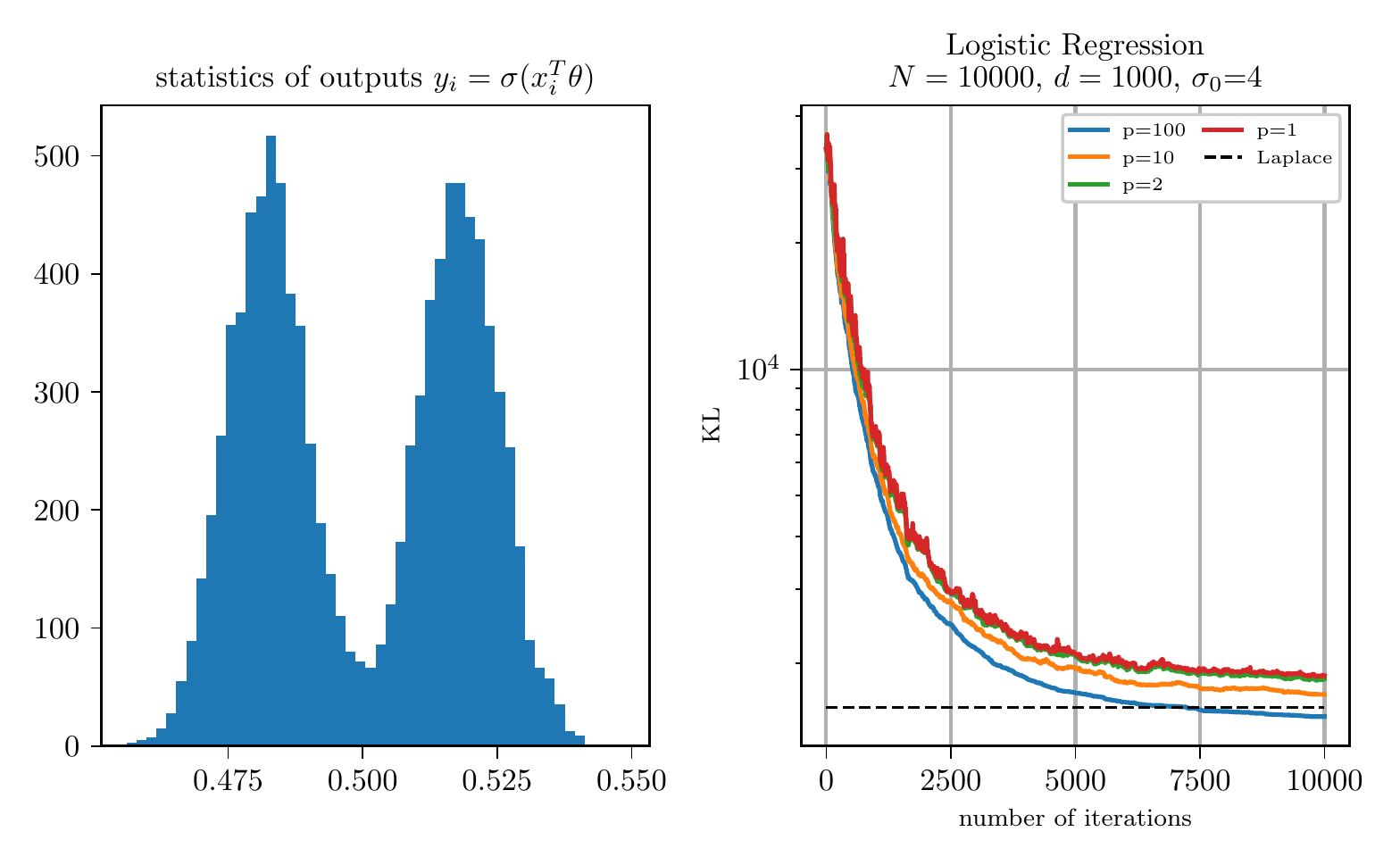}

\vspace*{-.2cm} 

\caption{Sensitivity to the latent dimension $p$. The L-RVGA algorithms may provide a lower divergence than the batch Laplace approximation in the sense of KL divergence to the true posterior. The estimated mean is roughly aligned with the maximum a posteriori (MAP). Only one inner loop proves sufficient to ensure convergence. The prior is set to be $P_0=\sigma_0 \mathbb{I}_d$ with $\sigma_0=4$ and the algorithm uses it for initialization.  The KL divergence is computed with sampling and is unnormalized.}
\label{figureXPLog1}
\end{figure}

\subsection{General nonlinear method evaluation}
\label{NonLin}
We address in this Section the general nonlinear L-RVGA, based on  the approximations detailed in Section~\ref{LRVGA general}. To assess the method, we choose to apply the general nonlinear method to logistic regression, again, as we have closed-form expressions that may serve as ground truth. 

As the generalized Gauss-Newton approximation, presented in Section \eqref{ggnsec}, is always exact in the logistic case, only the other approximations will be evaluated: sampling to approximate the expectation, extra-gradient to approximate the implicit scheme, and factor analysis to approximate the covariance matrix.

As logistic regression is based on a generalized linear model,  $\mu_t$ and $P_t$    are given by the RVGA updates for generalized linear models \eqref{rvgaGLM1}-\eqref{rvgaGLM2}. These updates involve  the expectation of the gradient and the Hessian of the logistic loss: 
\begin{align}
\mathbb{E}_{\theta}[\nabla_\theta \ell_t(\theta)]&=-y_tx_t + \mathbb{E}_{\theta}[\sigma(x_t^T\theta)]x_t   \label{ExpectLog1Z}\\
\mathbb{E}_{\theta}[\nabla^2_\theta \ell_t(\theta)]&=\mathbb{E}_{\theta}[\sigma(x_t^T\theta)(1-\sigma(x_t^T\theta))]x_tx_t^T   \label{ExpectLog2Z} .   
\end{align}  
Those expectations could be easily computed using equations \eqref{ExpectLog1}-\eqref{ExpectLog2}. For comparison purposes,  we may approximate those expectations  using the general sampling procedure of Section \ref{samplesec}, that is,
\begin{align}
&\mathbb{E}_{\theta}[\sigma(x_t^T\theta)]  \approx \frac{1}{K}  \sum_{i=1}^K \sigma(x_t^T\theta_i)  \label{is;eq}\\
&\mathbb{E}_{\theta}[\sigma(x_t^T\theta)(1-\sigma(x_t^T\theta))] \approx  \frac{1}{K}  \sum_{i=1}^K  \sigma(x_t^T\theta_i)(1-\sigma(x_t^T\theta_i)), \nonumber
\end{align} 
where the samples $\theta_i$ are drawn from the Gaussian distribution $\mathcal{N}(\mu,(\Psi + WW^T)^{-1})$ exactly using ensemble sampling as described in Section \ref{samplesec}. Table \ref{tableXPsampling} shows that the method approximates well the closed-form expression used in Section \ref{LogReg}. 

We then combined sampling with the extra-gradient method, where we skipped the extra covariance update, as described in Section \ref{extrasec} since we observed it helped convergence. We found few samples may be sufficient to provide good convergence of the algorithm in terms of KL divergence, as shown in Figure \ref{figureSampling2}. 

\begin{table}[!ht]
\centering
\begin{tabular}{|c|c|c|}
  \hline
 Method & $\mathrm{Tr} (WW^T+\Psi)^{-1}$ & $\mathbb{E}_{\theta}[ \sigma(x_t^T\theta)] $ \\
  \hline
Baseline & 9.24 & 0.7738\textbf{01}\\
Cholesky sampling & 9.66 & 0.7738\textbf{57}\\
\textbf{Ensemble sampling}  & 9.53 &  0.7738\textbf{55}\\
  \hline
\end{tabular}
\caption{\textbf{Test for the ensemble sampling approximation in dimension $d=10000$ and $p=10$}: 
Approximation of the expectation $\mathbb{E}_{\theta}[ \sigma(x_t^T\theta)] $ is performed  with $10$ samples. The closed-form expression with the inverse probit approximation is considered as the baseline. We also show how the trace of the matrix $(WW^T + \Psi)^{-1}$ is approximated. The Cholesky sampling uses the square root decomposition of $(WW^T+\Psi)^{-1}$. The ensemble sampling is our method described in Section  \ref{samplesec}.  
It achieves similar results as Cholesky sampling which is not memory-limited since it uses the full matrix and inverts it.  }
\label{tableXPsampling}
\end{table}

\begin{figure}[!ht]
\includegraphics[scale=0.78]{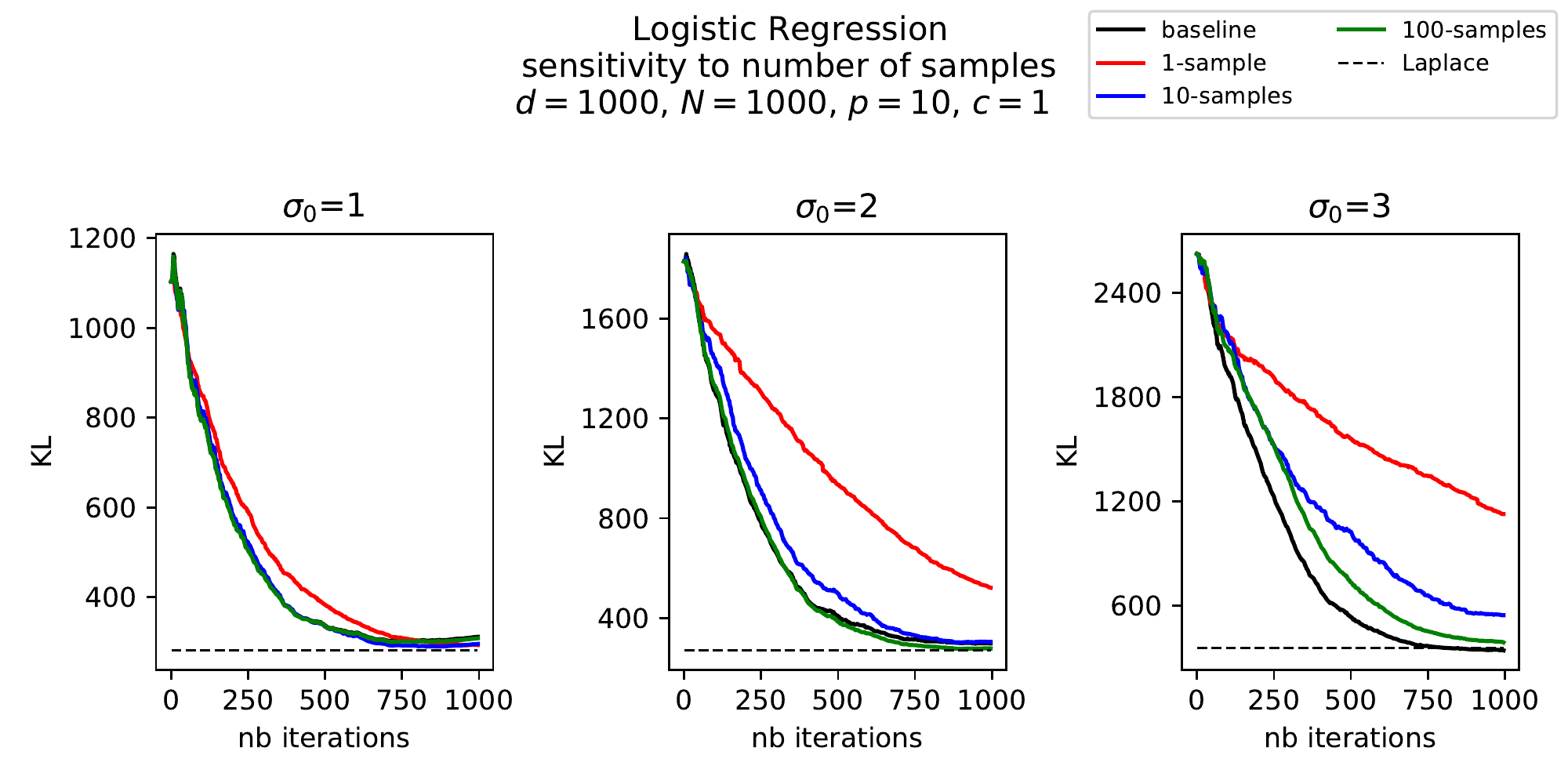}

\vspace*{-.2cm}

\caption{Approximation of the logistic posterior  for different deviations on the priors $\sigma_0=1,2$ and $3$. The inputs have been normalized on average. We compare the L-RVGA where the expectations are computed analytically (baseline) and the L-RVGA  where the expectations are approximated with sampling for $1$, $10$, and $100$ samples. All the algorithms use a Mirror prox (where we have skipped the extra covariance update) and use the same factor analysis approximation $p=10$. We see K = 10 samples are sufficient even in the more difficult case to converge to the batch Laplace KL.}
\label{figureSampling2}
\end{figure}

In the general nonlinear case, the L-RVGA algorithm can be advantageously used when the observations are in large dimensions and also in large numbers. A batch approach cannot be used on those problems where one can only afford multiple passes, or even a single pass, over the data set. Standard variational Bayesian methods based on a mean-field approach are not suitable for this class of problem where all observations are required at each iteration. \cite{Opper99} consider an online variational method using only one observation at each update. However, they used the right-side KL (moment matching) and lose the guarantee to increase an (evidence) lower bound. Approaches based on the stochastic maximization of the evidence lower bound \cite{Ong2018,SLANG} tackle the factor analysis problem in an online way using an adaptive step descent in contrast to the L-RVGA which doesn't use a step parameter. These algorithms may be more adaptable to the smoothness of the problem using step-tuning.  We anticipate that adding process noise and learning it online can provide adaptability, in particular, it can compensate for all the approximation errors we added in the process.

\paragraph{Sources:} The sources of the code are available on Github on the following repository: \url{https://github.com/marc-h-lambert/L-RVGA}.

\section*{Conclusion}
We have developed a new second-order algorithm, called L-RVGA, for online variational inference which scales to both large data sets and high dimensions. This algorithm is based on a two-stage variational problem that combines a variational Gaussian approximation followed by a factor analysis approximation of the inverse of the covariance matrix.   L-RVGA is able to estimate the mean and the covariance of the distribution of the latent parameters in a memory-efficient and parameter-free way and with only one pass through the data. We have tested it on linear and logistic regression problems and shown how to extend it to more general nonlinear problems with extra-gradients, memory-efficient sampling, and the generalized Gauss-Newton approximation. To build our generic algorithm we have introduced two new tools: a recursive EM algorithm for factor analysis which is parameter-free and faster than the online EM in this context; and a  sampler for Gaussian distribution with a structured precision matrix. 

Beyond variational inference, we anticipate the L-RVGA may prove useful for stochastic optimization keeping only the mean and using the covariance to build an adaptive learning step. This version could be compared to state-of-the-art algorithms in limited memory optimization such as Adagrad or even L-BFGS. We will investigate this direction in future work and will extend our algorithm using adaptive filtering techniques to further increase   its performance. We believe that a deeper connection between the communities of stochastic optimization and Kalman filtering may bring new ideas to tackle nonlinear stochastic problems. 

\subsection*{Acknowledgements}
This work was funded by the French Defence procurement agency (DGA) and by the French government under management of Agence Nationale de la Recherche as part of the “Investissements d’avenir” program, reference ANR-19-P3IA-0001(PRAIRIE 3IA Institute). We also acknowledge support from the European Research Council (grant SEQUOIA 724063).

\bibliographystyle{plainnat}
\bibliography{L-RVGA.bib}  

\appendix

\section{The recursive EM} 
\label{Appendix}

\subsection{Derivation of the fixed point equation for factor analysis with EM}
\label{AppendixFixedPointEM}
In this Section, we show that finding the factor analysis parameter with EM is equivalent to iterating through the fixed point equation  \eqref{FP-EM}. In the context of probabilistic principal component analysis \citet{Tipping1999}  have highlighted (in their Appendix $B$) how the EM can be rewritten as a fixed point equation. The following proof is an extension of this result to the factor analysis case.

In factor analysis, we want to approximate the empirical covariance matrix $\frac{1}{K}\sum_{i=1}^K v_iv_i^T$ with a \textquote{diagonal + low rank} structure $\Psi+WW^T$. In the classical EM approach, we introduce latent variables $z_i$ such that our samples can be rewritten $v_i=Wz_i + \varepsilon_i$ where $\varepsilon_i \sim \mathcal{N}(0,\Psi)$. We note $V=\begin{pmatrix} v_1 \cdots v_K\end{pmatrix}$ 
the sample matrix of size $d \times K$ such that $\frac{1}{K}\sum_{i=1}^K v_iv_i^T=\frac{1}{K}VV^T$ and  $Z=\begin{pmatrix} z_1  \cdots  z_K\end{pmatrix}$ the latent matrix of size $p \times K$.

At the expectation step (E-step), the latent variables $z_i$ are estimated conditionally on the observation and the current parameter estimates using the conditioning formula:
 \begin{align}
&p(z_i|v_i,W,\Psi)=\mathcal{N}(M^{-1}W^T\Psi^{-1}v_i,M^{-1}) \text{ where }   M=\mathbb{I}_p+W^T\Psi^{-1}W.
\end{align} 

At the maximization step (M-step), the latent variables are assumed fixed and the parameters are adjusted to maximize the total likelihood defined by:
\begin{align}
L(W,\Psi)=\log p(V,Z|W,\Psi)&=\sum_{i=1}^K \log p(v_i|z_i,W,\Psi) + \sum_{i=1}^K \log p(z_i) \\
&=-\frac{K}{2} (v_i-Wz_i)^T\Psi^{-1}(v_i-Wz_i)-\frac{K}{2}\log \det \Psi  -\frac{K}{2}  \log  (2 \pi) + \sum_{i=1}^K \log p(z_i).
\end{align}

The derivative of the expectation of $L$ with respect to $W$ gives:
\begin{align}
&\frac{\partial}{\partial W} \mathbb{E}_{Z \sim p(Z|V)}[L(W,\Psi)]= - \sum_{i=1}^K \Psi^{-1}W \mathbb{E}[z_iz_i^T|v_i] - \sum_{i=1}^K \Psi^{-1}v_i\mathbb{E}[z_i|v_i]^T=0,
\end{align}
and we obtain the optimum $W^{(n)}= \sum_{i=1}^K v_i\mathbb{E}[z_i|v_i]^T (\sum_{i=1}^K \mathbb{E}[z_tz_t^T|v_i])^{-1}$, where the $^{(n)}$ stands for ``new''.

Taking the derivative of the expectation of $L$ with respect to $\Psi^{-1}$ and keeping only the diagonal terms yields:
\begin{align}
\Psi^{(n)}&= {\rm{diag}}  (\frac{1}{K}\sum_{i=1}^K v_iv_i^T+\frac{1}{K} \sum_{i=1}^K W^{(n)} \mathbb{E}[z_iz_i^T|v_i]{W^{(n)}}^T -\frac{2}{K} \sum_{i=1}^K W^{(n)}\mathbb{E}[z_i|v_i] v_i^T) \label{eqDerivLik}\\
&= {\rm{diag}}  (\frac{1}{K}\sum_{i=1}^K v_iv_i^T-W^{(n)}\frac{1}{K} \sum_{i=1}^K \mathbb{E}[z_i|v_i] v_i^T),
\end{align}
where we have replaced ${W^{(n)}}^T$ by its optimal value in equation \eqref{eqDerivLik}.

And finally, the EM updates give:
\begin{align}
&\intertext{ \textbf{E-Step:}}
& \mathbb{E}[z_i|v_i]=M^{-1}W^T\Psi^{-1}v_i \nonumber \\
& \mathbb{E}[z_iz_i^T|v_i]=M^{-1}+\mathbb{E}[z_i|v_i]\mathbb{E}[z_i|v_i]^T.  \label{E-step2}
&\intertext{ \textbf{M-Step:}}
&W^{(n)}= \sum_{i=1}^K v_i \mathbb{E}[z_i|v_i]^T ( \sum_{i=1}^K \mathbb{E}[z_iz_i^T|v_i])^{-1} \nonumber\\
&\Psi^{(n)}={\rm{diag}} (\frac{1}{K}\sum_{i=1}^K v_iv_i^T-W^{(n)}\frac{1}{K} \sum_{i=1}^K \mathbb{E}[z_i|v_i]v_i^T). \label{M-step2}
\end{align}

The EM updates \eqref{E-step2} and \eqref{M-step2} can be rewritten in batch form using the matrix notation $V$ and $Z$:  
\begin{align}
&\intertext{ \textbf{E-Step:}}
&\mathbb{E}[Z|V]=M^{-1}W^T\Psi^{-1}V \nonumber\\
& \sum_{i=1}^K \mathbb{E}[z_iz_i^T|v_i]:=\mathbb{E}[ZZ^T|V]=KM^{-1}+\mathbb{E}[Z|V]\mathbb{E}[Z|V]^T=K(M^{-1}+M^{-1}W^T\Psi^{-1}S_K\Psi^{-1}WM^{-1}).
&\intertext{ \textbf{M-Step:}}
&W^{(n)}=V \mathbb{E}[Z|V]^T \mathbb{E}[ZZ^T|V]^{-1} \nonumber\\
&\Psi^{(n)}= {\rm{diag}}  (\frac{1}{K}VV^T-W^{(n)}\frac{1}{K} \mathbb{E}[Z|V]V^T).
\end{align}

And finally the E-step and M-step can be fused to form a fixed-point equation:

\begin{align*}
W^{(n)}&=V \mathbb{E}[Z|V]^T \mathbb{E}[ZZ^T|V]^{-1} \\
&=S_K \Psi^{-1}W M^{-1}(M^{-1}+M^{-1}W^T\Psi^{-1}S_K\Psi^{-1}WM^{-1})^{-1}\\
&=S_K \Psi^{-1}W (\mathbb{I}_p+M^{-1}W^T\Psi^{-1}S_K\Psi^{-1}W)^{-1}\\
\Psi^{(n)}&=\rm{diag} (\frac{1}{K}VV^T-W^{(n)}\frac{1}{K} \mathbb{E}[Z|V]V^T)\\
&=\rm{diag} (S_K-W^{(n)}M^{-1}W^T\Psi^{-1}S_K).
\end{align*}

which is the fixed point equation given in \eqref{FP-EM}.

\section{The fixed point EM is equivalent to the MLE fixed point}
\label{FixedPoint}
In this Section we show new results concerning the equivalence of the fixed points for the marginal likelihood (MLE algorithm) and the total likelihood (EM algorithm) for the particular case of factor analysis. We consider here the batch factor analysis problem. This result may not be applied to our RVGA algorithm but to each inner loop which is guaranteed to increase the one sample likelihood. 

The marginal likelihood and the total likelihood are both related as follows: 
\begin{align}
&\underset{W,\Psi}{\max}  \quad \log p(v_1, \dots, v_N|W,\Psi) \text{ where } S_N=\frac{1}{N}\sum_{i=1}^Nv_iv_i^T=P^{-1} \quad \text{ (MLE) }.\label{marg:li}\\
&\leq \underset{W,\Psi}{\max} \quad  \mathbb{E}_z[\log p(v_1, \dots, v_N,z|W,\Psi)] \text{ with } p(v\mid z)\sim\mathcal{N}(Wz,\Psi),~ p(z)\sim \mathcal{N}(0,I_p)\quad \text{ (EM) }.\label{tot:li}
\end{align}

The advantage of the EM approach is that the total likelihood appearing in  \eqref{tot:li} is guaranteed to increase at each step, i.e.,  the algorithm is stable. The EM algorithm may not always converge to the maximum of the marginal likelihood appearing in \eqref{marg:li} but converges to a stationary point which turns out to be also a stationary point for the maximum likelihood. This fact was already highlighted by \citet{incEM} and \citet{onlineEM} but we specify the result in the case of factor analysis in an algebraic way. We first write the maximum likelihood  as a fixed point, then show an equivalence between the two fixed points. This result is finally used to show an equivalence between different eigenvalues decomposition algorithms. 

\subsection{Factor analysis with maximum likelihood (MLE) as a fixed point.} 
\label{FAMLEfixedPoint}
\label{AppendixStandardFA}
In this section, we write the maximum likelihood over the parameters of the factor analysis in the form of a fixed-point equation. This is a well-known result derived from chapter $21.2$ of the book of \citet{barber2011} but we want to use this result to make a connection with the fixed point obtained with EM. We want to approximate a matrix $S$ with a matrix $WW^T+\Psi$ by maximizing the following likelihood (where $C$ is a constant): 
\begin{align}
&\underset{W,\Psi}{\max}  \quad L(W,\Psi)= \underset{W,\Psi}{\max} \quad -\frac{N}{2}  \mathrm{Tr}((WW^T+\Psi)^{-1}S_N) - \frac{N}{2}  \log \det (WW^T+\Psi) + C.
\intertext{A necessary condition on the optimal solution is to zeroing the gradients:  }
&\frac{\partial L(W,\Psi)}{\partial W} = (WW^T+\Psi)^{-1}W^T-(WW^T+\Psi)^{-1}S_N(WW^T+\Psi)^{-1}W=0 \\
&\frac{\partial L(W,\Psi)}{\partial \Psi}=(WW^T+\Psi)^{-1}-(WW^T+\Psi)^{-1}S_N(WW^T+\Psi)^{-1}=0,
\intertext{leading to the fixed-point equations:}
&\textbf{Fixed-point equations for the MLE} \\
&W^{(n)}=S_N(WW^T+\Psi)^{-1}W \label{eqFAappendix2} \\
&\Psi=\rm{diag}(S_N-W^{(n)}{W^{(n)}}^T).
\end{align}
 
\subsection{Equivalence of fixed points and convergence}
\label{FixedPoints}
The following proposition shows the algebraic  equivalence of the fixed points.
\begin{proposition}
The factor analysis parameters  which maximize the marginal likelihood, that is, the solution to  \eqref{marg:li} satisfy  the following fixed point equation: 
\begin{align}
W^{(n)}&=S_N(WW^T+\Psi)^{-1}W \\
\Psi&={\rm{diag}}(S_N-W^{(n)}{W^{(n)}}^T).
\intertext{This fixed point equation is equivalent to the EM fixed-point equation:}
W^{(n)}&=S_N \Psi^{-1}W (\mathbb{I}_p+M^{-1}W^T\Psi^{-1}S_N\Psi^{-1}W)^{-1}\\
&\text{where: } M=\mathbb{I}_p+W^T\Psi^{-1}W \\
\Psi&={\rm{diag}} (S_N-W^{(n)}M^{-1}W^T\Psi^{-1} S_N).
\end{align}As a consequence, the EM algorithm converges to a stationary point which is also a stationary point for the  likelihood. \label{fix:prop}
\end{proposition}
 
\begin{proof}
The proof for the equivalence of the fixed points is straightforward,  the update for $W$ can be rewritten as:
\begin{align}
W&=S_N(WW^T+\Psi)^{-1}W\\
&=S_N\Psi^{-1}W(\mathbb{I}_p+W^T\Psi^{-1}W)^{-1} \text{ (From the Woodbury formula)} \label{FP1}\\
&=S_N\Psi^{-1}WM^{-1}  \label{FP2}\\
&=S_N\Psi^{-1}W(\mathbb{I}_p+(S_N\Psi^{-1}WM^{-1})^T\Psi^{-1}W)^{-1}\\
&=S_N \Psi^{-1}W (\mathbb{I}_p+M^{-1}W^T\Psi^{-1}S_N\Psi^{-1}W)^{-1},\\
\intertext{where we have replaced the term $W^T$ in \eqref{FP1} by its development in \eqref{FP2}.}
\intertext{The update for $\Psi$ can be rewritten as:}
\Psi&=\rm{diag} (S_N-WW^T) \label{FP3}\\
&={\rm{diag}} (S_N-W(S_N\Psi^{-1}WM^{-1} )^T)\\
\intertext{where we have replaced the term $W^T$ in \eqref{FP3} by its development in \eqref{FP2}, leading to} 
\Psi&={\rm{diag}} (S_N-WM^{-1}W^T\Psi^{-1} S_N) \\
&={\rm{diag}} (S_N-W_nM^{-1}W^T\Psi^{-1} S_N).
\end{align}
\end{proof}

\subsection{Relation to singular value decomposition}

The fixed point equation from the maximum likelihood is solved using a singular value decomposition (SVD) of $S_N$ (see \citet{barber2011}, chapter $21.2$). The equivalence with the EM fixed point suggests that the EM make implicitly an SVD decomposition. It can be shown it is the case if we consider an asymptotically Probabilistic Principal Component analysis form (PPCA), ie $\Psi=\sigma I$ where we let tends the parameter $\sigma$ to $0$. This result was shown by \citet{Roweis97} for the fixed point EM with PPCA and is related to the MLE fixed point in the following Corollary.

\begin{corollary}[of Prop. \ref{fix:prop}]
\begin{align}
\intertext{The factor analysis MLE fixed-point equation  for PPCA:}
W^{(n)}&=S_N(WW^T+\sigma^2 \mathbb{I}_d)^{-1}W=S_NW(\sigma^2 \mathbb{I}_p+W^TW)^{-1}  
\intertext{converge for $\sigma \rightarrow 0$ to a fixed point :}
W^{(n)}&=S_NW(W^TW)^{-1},
\intertext{which corresponds to the power method:}
w& \leftarrow S_N\frac{w}{||w||^2} \quad \text{  (given here in vectorial form)}.
\intertext{The factor analysis EM fixed-point for PPCA:}
W_n&=S_N W (\sigma^2 \mathbb{I}_p+(\sigma^2 \mathbb{I}_p+W^TW)^{-1}W^TS_NW)^{-1} 
\intertext{converge for $\sigma \rightarrow 0$ to a fixed point :}
W^{(n)}&=S_N W (W^TS_NW)^{-1}W^TW,
\intertext{which is equivalent to the EM-PCA  method \citep{Roweis97}:}
X&=(W^TW)^{-1}W^TY\\
W^{(n)}&=YX^T(XX^T)^{-1} \text{ where  $Y$ is defined such that } S_N=YY^T.
\end{align}
\end{corollary}

\begin{proof}
The proof is direct.
\end{proof}

\section{Recursive factor analysis of the covariance matrix}\label{FAforMatrix} 
The retained factorization $W_tW_t^T+\Psi_t$ of the precision matrix $P_t^{-1}$ associated with the Bayesian parameter $\theta_t$ is not conventional in factor analysis, which is usually applied to the empirical covariance matrix $S_t=\frac{1}{t} \sum_{i=1}^t x_ix_i^T$ of the inputs $x_t$.  
In the linear case, both can be related as follows $S_t :=\frac{1}{t}P_t^{-1}$ for $t >0$, and at $t=0$ we let $S_0=P_0^{-1}$. $S_t$ can be expressed in a recursive way as:
\begin{align}
&S_t=\frac{1}{t} P_t^{-1}=\frac{t-1}{t}\frac{P_{t-1}^{-1}}{t-1} +\frac{1}{t} x_tx_t^T=\frac{t-1}{t}S_{t-1}+\frac{1}{t} x_tx_t^T. \label{CovEmp}
\end{align}
If we let $S_0$ be null, we obtain exactly $\frac{1}{t} \sum_{i=1}^t x_ix_i^T$, otherwise, we obtain a regularized version of the empirical covariance. 
The corresponding recursive factor analysis form is: 
\begin{align}
&S_t=\frac{t-1}{t}(W_{t-1}W_{t-1}^T+\Psi_{t-1})+\frac{1}{t} x_tx_t^T,  \label{CovEmpFA}
\end{align}
which fits into the problem \eqref{NotNormalizedRecFA2}. It may thus be addressed through    Algorithm \ref{algorithmRecEM} with $\alpha_t=(t-1)/t$ and $\beta_t=1/t$. We may guess the value of $S_0$ to initialize the procedure and make it more stable. Observing that $\mathrm{Tr} \frac{1}{N} \sum_{t=1}^N x_tx_t^T=\frac{1}{N} \sum_{t=1}^N ||x_t||^2$ we see that the trace of the unknown covariance must match the expectation of the square norm of inputs. This expectation may be estimated on a batch of size $M$. We use this property to compute $S_0$ as follows: get a batch formed by the first $M$ data inputs   $x_1,\dots,x_M$ , and let  $S_0=\frac{1}{\sigma_0^2} \mathbb{Id}$ with $\sigma_0=\sqrt{\frac{d}{ \frac{1}{M}  \sum_{t=1}^M ||x_t||^2}}$. This is equivalent to normalizing the inputs in mean and set $\sigma_0=\sqrt{d}$. 

\section{Initialization and prior information}\label{initPrior}
Given a prior covariance $P_0$ on  $\theta$, we want to initialize $W_0$ and $\Psi_0$ such that $W_0W_0^T+\Psi_0 \approx P_0^{-1}$. 
We must suppose $P_0$ is sparse enough to fit in memory and has an inverse which can be approximated by $W_0W_0^T+\Psi_0$ for example using the Woodbury formulas. For the experiments, we consider an isotropic initial covariance $P_0=\sigma_0^2 \mathbb{I}_d$.  We then compute $W_0$ and $\Psi_0$ such that  $W_0W_0^T+\Psi_0 = \frac{1}{\sigma_0^2} \mathbb{I}_d$.  The simple choice $\Psi_0 =\frac{1}{\sigma_0^2} \mathbb{I}_d$ and $W_0=0_{d\times p}$ would make the algorithm run into problems as $W_0=0_{d\times p}$  is a  stationary point of the fixed-point equation \eqref{FP-EM}.

We use the following rule $\Psi_0=\psi_0 \mathbb{I}_d$ where $\psi_0>0$ is a scalar and generate $W_0$ as a $d\times p$ matrix whose columns are the vectors $u_0^1,u_0^2,\dots,u_0^p$ independently drawn from an isotropic Gaussian distribution in $\mathbb R^d$ and which have been normalized so that $\forall k: ||u_0^k||=w_0$. We then let:
$$
\psi_0=(1-\varepsilon) \frac{1}{\sigma_0^2},\qquad w_0=\sqrt{\frac{\varepsilon d}{p}}  \frac{1}{\sigma_0},
$$with $0<\varepsilon\ll 1$ a small parameter. The rationale is that   $W_0W_0^T = \sum_{k=1}^p u_0^k {u_0^k}^T$ so that we have:
$$ \mathrm{Tr} (W_0W_0^T+\Psi_0)=  \sum_{k=1}^p \mathrm{Tr} (u_0^k {u_0^k}^T)+\psi_0 \mathrm{Tr} \mathbb{I}_d=p w_0^2+d\psi_0=\frac{d}{\sigma_0^2}=\mathrm{Tr}P_0^{-1}.$$

\begin{remk}
This initialization can be extended to the case where $P_0$ is a diagonal matrix $P_0=\mathrm{Diag}(\sigma^2_1,\dots,\sigma^2_d)$ where $\sigma_i^2$ represents now a variance on the $i^{th}$ coordinate. 
\end{remk}

\section{Derivation of the online EM algorithm for factor analysis}
\label{onlineEM} 
In this section, we derive the online EM algorithm  \citep{onlineEM} for the factor analysis problem to compare it to our Recursive EM algorithm. We use the same notation as in the previous section where $v_1,\dots,v_N$ are our $N$ observations in dimension $d$ and $z_1,\dots,z_N$ are our latent variables in dimension $p$. Moreover, we suppose that each couple of variables $(z_t,v_t)$ are independent and belong to an exponential family given by $\log p(z_t,v_t|\theta)=<S(z_t,v_t),\phi(\theta)>-F(\theta)$, where  $F$ is the log partition function, $\phi$ is a function which map the natural parameter and $S(z_t,v_t)$ are the sufficient statistics.

The online EM algorithm \citep{onlineEM} considers the following  fixed point equation with respect to the sufficient statistics $S$:
\begin{align}
&S=\mathbb{E}_{v_t} \mathbb{E}_{z_t \sim p(z|v_t,\theta^*(S))}[S(z_t,v_t)]=T(S) \\
&\text{ where } \theta^*(S)=\underset {\theta \in \Theta}{\arg \max} \sum_{t=1}^N <S(z_t,v_t),\phi(\theta)>-F(\theta),
\end{align}
and solve $T(S)-S$ using a stochastic root solver based on the Robbins–Monro algorithm \citep{robbins1951}. The sufficient statistics $S$ and the optimal parameter $\theta$ are update online with an adaptive step $\gamma_t$ at each new incoming observations $v_t$ as  follows:
\begin{align}
& S_t = (1-\gamma_t) S_{t-1} + \gamma_t \mathbb{E}_{z_t \sim p(z|v_t,\theta_{t-1})}[S(z_t,v_t)] \quad \text{ (E-step)}\\
& \theta_t=\underset {\theta \in \Theta}{\arg \max} <S_t,\phi(\theta)>-F(\theta) \quad \text{ (M-step)}.
\end{align}

In the case of factor analysis, the joint distribution is:
\begin{align}
\log p(v_t,z_t) = -\frac{1}{2} \mathrm{Tr}[\Sigma^{-1}S(v_t,z_t)] -\frac{1}{2} \log \det \Sigma +c,
\end{align}
where the joint covariance is :
\begin{align}
\Sigma=\begin{pmatrix} WW^T+\Psi & W \\ W^T & I_p \end{pmatrix},
\end{align}
and the sufficient statistics are :
\begin{align}
S(v_t,z_t)=\begin{pmatrix} v_t \\ z_t \end{pmatrix}  \begin{pmatrix} v_t \\ z_t \end{pmatrix}^T=\begin{pmatrix} v_t v_t^T & v_t z_t^T \\ z_tv_t^T & z_tz_t^T \end{pmatrix}.
\end{align}

The expectation of sufficient statistics  gives 
\begin{align}
\mathbb{E}_{z_t \sim p(z|v_t,\theta_{t-1})}[S(v_t,z_t)] &=\begin{pmatrix} v_t v_t^T & v_t \mathbb{E}[z_t^T|v_t]   \\ \mathbb{E}[z_t|v_t]v_t^T & \mathbb{E}[z_tz_t^T|v_t] \end{pmatrix} 
\end{align}

Rather than update the full matrix $S(v_t,z_t)$ we will update the blocks: $S_{1.t}=v_t v_t^T$, $S_{2.t}=\mathbb{E}[z_t|v_t]v_t^T$ and $S_{3.t}=\mathbb{E}[z_tz_t^T|v_t]$, which are necessary to compute the M-step. 

Finally, using the same notation as in the previous Section, the online EM updates for factor analysis become:
\begin{equation}
\label{StochEM}
\begin{aligned}
& \textbf{Online EM}\\
& \textbf{ E-step}\\
S_{1.t}&=(1-\gamma_t)  S_{1.t-1} + \gamma_t v_tv_t^T\\
S_{2.t}&= (1-\gamma_t)  S_{2.t-1} + \gamma_t  \mathbb{E}[z_t|v_t]v_t^T \\
&=(1-\gamma_t)  \mu_{t-1} + \gamma_t M_{t-1}^{-1}W_{t-1}^T\Psi_{t-1}^{-1}v_tv_t^T \\
S_{3.t}&= (1-\gamma_t)  S_{3.t-1} + \gamma_t \mathbb{E}[z_tz_t^T|v_t] \\
&= (1-\gamma_t)  S_{3.t-1} + \gamma_t (M_{t-1}^{-1}+M_{t-1}^{-1}W_{t-1}^T\Psi_{t-1}^{-1}v_t v_t^T\Psi_{t-1}^{-1}W_{t-1}M_{t-1}^{-1})\\
& \textbf{ M-step} \\
W_t&=  S_{2.t}^T S_{3.t}^{-1} \\
\Psi_t&=\rm{diag} (S_{1.t})-\rm{diag} (W_t S_{2.t}).
\end{aligned}
\end{equation}

To develop a limited memory version of this algorithm, we store only the diagonal of the high dimensional squared matrix $S_{1.t}$ and update it as follows:
\begin{align}
& \rm{diag} (S_{1.t})=(1-\gamma_t)\rm{diag} (S_{1.t-1}) + \gamma_t u_{t}*u_{t},
\end{align}
where $x*y$ is a component wise operation giving $x*y=\rm{diag}(xy^T)$ for two vectors. 

To choose the step size $\gamma_t$ we must satisfy the Robins Monro rules: 
\begin{align}
& \sum_{t=1}^N \gamma_t=\infty \quad \sum_{t=1}^N \gamma_t^2 <\infty.
\end{align}

We consider the following step recommended by \citet{onlineEM}:
\begin{align}
& \gamma_0 =1 \\
& \gamma_t=\frac{1}{t^{0.6}}\quad \forall  t > 0.
\end{align}

Finally, in post-processing, we use a Polyak-Ruppert halfway averaging to improve the convergence as recommended by \citet{onlineEM}:
\begin{align}
& \forall t > N/2 \quad st \quad \tilde{t}=t-N/2 >0 \quad do: \nonumber \\
&\bar{W}_t=\frac{\tilde{t}-1}{\tilde{t}} W_{t-1}+ \frac{1}{\tilde{t}} W_t \\
&\bar{\Psi}_t=\frac{\tilde{t}-1}{\tilde{t}} \Psi_{t-1}+ \frac{1}{\tilde{t}} \Psi_t.
\end{align}

\section{The outer product approximation in the general case.} 
\label{ProofLemmaGGN}
We show in this Section the relation between the generalized Gauss-Newton approximation and the Fisher matrix. The proof proposed here comes from \citet{ollivier2018}[Appendix $A$]:
 \begin{proof}
Let’s $p$ be an exponential family such that:
\begin{align}
&p(y|\theta)=m(y_t)\exp({\eta}^TT(y)-A( {\eta})),
\end{align} 
where $T(y)$ is the sufficient statistics, $A$ is the log partition function which satisfies $\nabla^2_\eta  A( {\eta})=Cov(y|\theta)$ and $\nabla_\eta  A( {\eta})=m=\mathbb{E}[(y|\theta)]$ and finally $\eta$ is the natural parameter which depends on $\theta$ through a function $h$, ie $\eta=h(\theta)$. Using twice the chain rules, the second derivative of the negative likelihood of $p$ writes  :
\begin{align}
-\frac{\partial^2 \ln p(y|\theta) }{\partial \theta^2}&= -\frac{\partial h}{\partial \theta} \frac{\partial^2 \ln p(y|\theta)}{\partial h^2} \frac{\partial h}{\partial \theta}^T - \frac{\partial^2 h}{\partial \theta^2} \frac{\partial \ln p(y|\theta)}{\partial h} \\
&= \frac{\partial h}{\partial \theta} Cov(y|\theta) \frac{\partial h}{\partial \theta}^T - \frac{\partial^2 h}{\partial \theta^2}(T(y)-m).
\end{align}
Taking the expectation under $y$ on both sides, we obtain directly the relation:

\begin{align}
\mathbb{E}_y[-\frac{\partial^2 \ln p(y/\eta) }{\partial \theta^2}]=\mathbb{F}(\theta)=\frac{\partial \eta}{\partial \theta} Cov(y|\theta) \frac{\partial \eta}{\partial \theta}^T.
\end{align}

And finally:
\begin{align}
\mathbb{E}_\theta[\mathbb{F}(\theta)]=\mathbb{E}_\theta[\frac{\partial h}{\partial \theta} \mathrm{Cov}(y|\theta) \frac{\partial h}{\partial \theta}^T].
\end{align}

Now for a generalized linear model such that $h(\theta)=x_t^T\theta$ the GGN approximation is exact since $\frac{\partial^2 h}{\partial \theta^2}=0$, this completes the proof.
\end{proof}

\section{Mirror prox}  
\label{Extragrad}
In this section, we show that the iterated scheme defined equation \eqref{mirrorProx} is equivalent to a Mirror prox update \citep{MirrorProx}. We recall first the connection between the recursive variational scheme and the Mirror descent using the results initially derived by \citet{khan2018} for the batch variational approach and extended by \citet{RVGA} for the recursive variational approach. 

Considering an exponential family $q_\eta$ of natural parameter $\eta$, mean parameter $m$ and a strictly convex log partition function $F$ such that $q_\eta(\theta)=h(\theta)\exp(<\eta,\theta>-F( {\eta}))$, the recursive  variational approximation problem between a target distribution  $q_\eta$ and the one-sample posterior $p(\theta|y_t) \propto p(y_t|\theta)q_{\eta_{t-1}}(\theta)$ writes:
\begin{align}
&\underset {\eta_t}{\arg \min} \quad KL(q_{\eta_{t}}(\theta)|p(\theta|y_t)) \\ \label{KLdiv}
=& \underset {\eta_t}{\arg \min} \quad  \mathbb{E}_{q_{\eta_{t}}}[-\log p(y_t|\theta)] + B_F(\eta_{t-1},\eta_{t}),
\end{align}
where $B_F$ is the Bregman divergence associated with the strictly convex log partition function $F$. The critical point must satisfy: 
\begin{align}
&\nabla_{\eta_t} \mathbb{E}_{q_{\eta_{t}}}[-\log p(y_t|\theta)]+(\eta_{t}-\eta_{t-1})\nabla^2 F(\eta_t) =0,
\end{align}
which gives the following implicit fixed point equation on the natural parameter:
\begin{align}
\eta_{t} &=\eta_{t-1} +(\nabla^2 F(\eta_t))^{-1} \nabla_{\eta} \mathbb{E}_{q_{\eta}}[\log p(y_t|\theta)] (\eta_t)\label{RVEA} \\
&=\eta_{t-1} + \nabla_{m} \mathbb{E}_{q_{m}}[\log p(y_t|\theta)](m_t).
\end{align}
If we consider the function $f(m)=\mathbb{E}_{q_m}[-\log p(y_t|\theta)]$ and use the relation $\eta=\nabla F^*(m)$, this update can be interpreted as an implicit version of a Mirror descent on $f$ with step size one: 
\begin{align}
&\nabla F^*(m_t)=\nabla F^*(m_{t-1})- \nabla_m f (m_{t}) \text{ (dual descent with step $1$) } \\
&m_t=\nabla F(\eta_t). \text{ ( projection) }
\end{align}
This implicit scheme can be approximated using the Mirror prox algorithm \citep{MirrorProx} with a step size one:
\begin{align*}
 &\textbf{Mirror prox} \\
& \hat{\eta}_t=\nabla F^*(\hat{m}_t)=\nabla F^*(m_{t-1})- \nabla_m f (m_{t-1}) \\
& \hat{m}_t=\nabla F(\hat{\eta}_t) \\
& \eta_t=\nabla {F^*}(m_t)=\nabla {F^*}(m_{t-1})- \nabla_m f (\hat{m}_t) \\
& m_t=\nabla F(\eta_t). 
\end{align*}

In fact, we have rather computed here a stochastic version of Mirror prox. The stochastic version of Mirror prox may inherit the good convergence properties of the original batch version if the function $f$ is convex and the step is adaptive \citep{MirrorProxStoch}. Here the setting is different: our function $f$ is not jointly convex in $\mu$ and $P$ and we consider a constant step of size one. However, we have found the stochastic version behaves empirically well. 

In the case where $q$ is a multivariate Gaussian distribution, the mean and natural parameters are given by $\eta=\nabla F^*(m)=\begin{pmatrix} \eta_1=P^{-1}\mu  \\ \eta_2=-\frac{1}{2}P^{-1}  \end{pmatrix}$ and $m=\nabla F(\eta)=\begin{pmatrix} m_1=\mu \\ m_2=P+\mu \mu^T \end{pmatrix}$.

The gradient with respect to the mean parameters $m_1,m_2$ can be expressed as the gradient with respect to the source parameters $\mu,P$ using the chain rule:
\begin{align}
&\frac{\partial f}{\partial m_1}=\frac{\partial f}{\partial \mu}\frac{\partial \mu}{\partial m_1}+\frac{\partial f}{\partial P}\frac{\partial P}{\partial m_1}=\frac{\partial f}{\partial \mu}-2\frac{\partial f}{\partial P}\mu\\
&\frac{\partial f}{\partial m_2}=\frac{\partial f}{\partial \mu}\frac{\partial \mu}{\partial m_2}+\frac{\partial f}{\partial P}\frac{\partial P}{\partial m_2}=\frac{\partial f}{\partial P}
\end{align}
A step of the Mirror prox update:
\begin{align}
 \nabla F^*(m_t)&=\nabla F^*(m_{t-1})- \nabla_m f (m_{t-1}) \\
\iff \eta_t&=\eta_{t-1} - \nabla_m f (m_{t-1}),
\end{align}
become, if we write $f(m(\mu,P))=\mathbb{E}_{\theta \sim \mathcal{N}(\mu,P)}[\log p(y_t|\theta)]$:

\begin{align}
&\begin{pmatrix}P_t^{-1}\mu_t \\ -\frac{1}{2}P_t^{-1} \end{pmatrix}=\begin{pmatrix} P_{t-1}^{-1}\mu_{t-1} \\ -\frac{1}{2}P_{t-1}^{-1}  \end{pmatrix} + \begin{pmatrix} \frac{\partial f}{\partial \mu} \lvert_{\substack{\mu_{t-1},P_{t-1}}} -2\frac{\partial f}{\partial P} \lvert_{\substack{\mu_{t-1},P_{t-1}}} \mu_{t-1}\\ \frac{\partial f}{\partial P} \lvert_{\substack{\mu_{t-1},P_{t-1}}} \end{pmatrix} \\
\iff &\begin{pmatrix}P_t^{-1}\mu_t \\ -\frac{1}{2}P_t^{-1} \end{pmatrix}=\begin{pmatrix} (P_{t-1}^{-1}-2\frac{\partial f}{\partial P} \lvert_{\substack{\mu_{t-1},P_{t-1}}}  \mu_{t-1} \\ -\frac{1}{2}P_{t-1}^{-1}  \end{pmatrix} + \begin{pmatrix} \frac{\partial f}{\partial \mu}  \lvert_{\substack{\mu_{t-1},P_{t-1}}}  \\ \frac{\partial f}{\partial P}  \lvert_{\substack{\mu_{t-1},P_{t-1}}}  \end{pmatrix} \\
\iff &\begin{pmatrix}P_t^{-1}\mu_t \\ -\frac{1}{2}P_t^{-1} \end{pmatrix}=\begin{pmatrix} P_t^{-1} \mu_{t-1} \\ -\frac{1}{2}P_{t-1}^{-1}  \end{pmatrix} + \begin{pmatrix} \frac{\partial f}{\partial \mu}  \lvert_{\substack{\mu_{t-1},P_{t-1}}}  \\ \frac{\partial f}{\partial P}  \lvert_{\substack{\mu_{t-1},P_{t-1}}}  \end{pmatrix} \\
\iff &\begin{pmatrix}P_t^{-1}\mu_t \\ -\frac{1}{2}P_t^{-1} \end{pmatrix}=\begin{pmatrix} P_t^{-1} \mu_{t-1} \\ -\frac{1}{2}P_{t-1}^{-1}  \end{pmatrix} + \begin{pmatrix} -\mathbb{E}_{\theta \sim \mathcal{N}(\mu_{t-1},P_{t-1})}[\nabla_\theta \log p(y_t|\theta)] \\ \frac{1}{2} \mathbb{E}_{\theta \sim \mathcal{N}(\mu_{t-1},P_{t-1})}[\nabla^2_\theta \log p(y_t|\theta)] \end{pmatrix}. \label{BonnetPrice}
\end{align}

The last derivation \eqref{BonnetPrice}  comes from the Bonnet \& Price formulas \citep{lin2019b}:  

\begin{align}
& \nabla_\mu \mathcal{N}(\theta|\mu,P)=-\nabla_\theta \mathcal{N}(\theta|\mu,P) \\
& \nabla_P\mathcal{N}(\theta|\mu,P)=\frac{1}{2}\nabla_\theta^2 \mathcal{N}(\theta|\mu,P).
\end{align}

Rearranging terms and applying two times the update as in Mirror prox gives the iterated scheme defined in equation \eqref{mirrorProx}:
 \begin{align}
&\mathbf{\hat{P}_{t}}^{-1}=P_{t-1}^{-1} - \mathbb{E}_{\theta \sim \mathcal{N}(\mu_{t-1},P_{t-1})}[\nabla^2_{\theta} \log p(y_t|\theta)] \label{MP-1} \\
&  \mathbf{\hat{\mu}_{t}} =\mu_{t-1}+ \mathbf{\hat{P}_{t}} \mathbb{E}_{\theta \sim \mathcal{N}(\mu_{t-1},P_{t-1})}[\nabla_{\theta} \log p(y_t|\theta)] \label{MP-2} \\
&\mathbf{P_{t}}^{-1}=P_{t-1}^{-1} - \mathbb{E}_{\theta \sim \mathcal{N}( \mathbf{\hat{\mu}_t}, \mathbf{\hat{P}_{t}})}[\nabla^2_{\theta} \log p(y_t|\theta)] \label{MP-3} \\
&\mu_{t}=\mu_{t-1} + \mathbf{P_{t}} \mathbb{E}_{\theta \sim \mathcal{N}( \mathbf{\hat{\mu}_t}, \mathbf{\hat{P}_{t}} )}[\nabla_{\theta} \log p(y_t|\theta)]. \label{MP-4} 
\end{align}

If the expectations are replaced with a linearization around the last estimated, these updates are also equivalent to the extended iterated Kalman filter scheme \citep{jazwinski1970}. 

Applying the mirror-prox scheme to our logistic regression problem \ref{LogReg} without factor analysis, we see that the Gaussian well approximates the logistic posterior in figure \ref{figureXP_ExtraGrad}. However, when we combine mirror-prox with factor analysis,  the extra covariance update can make the mirror-prox  scheme unstable. We have observed it is then  preferable   to skip the extra covariance update \ref{MP-3}, i.e. using:
 \begin{align}
&\mathbf{\hat{P}_{t}}^{-1}=P_{t-1}^{-1} - \mathbb{E}_{\theta \sim \mathcal{N}(\mu_{t-1},P_{t-1})}[\nabla^2_{\theta} \log p(y_t|\theta)] \\
&  \mathbf{\hat{\mu}_{t}} =\mu_{t-1}+ \mathbf{\hat{P}_{t}} \mathbb{E}_{\theta \sim \mathcal{N}(\mu_{t-1},P_{t-1})}[\nabla_{\theta} \log p(y_t|\theta)] \\
&\mu_{t}=\mu_{t-1}+ \mathbf{\hat{P}_{t}} \mathbb{E}_{\theta \sim \mathcal{N}( \mathbf{\hat{\mu}_t}, \mathbf{\hat{P}_{t}} )}[\nabla_{\theta} \log p(y_t|\theta)]. 
\end{align}
 
\begin{figure}[!ht]
\includegraphics[scale=1]{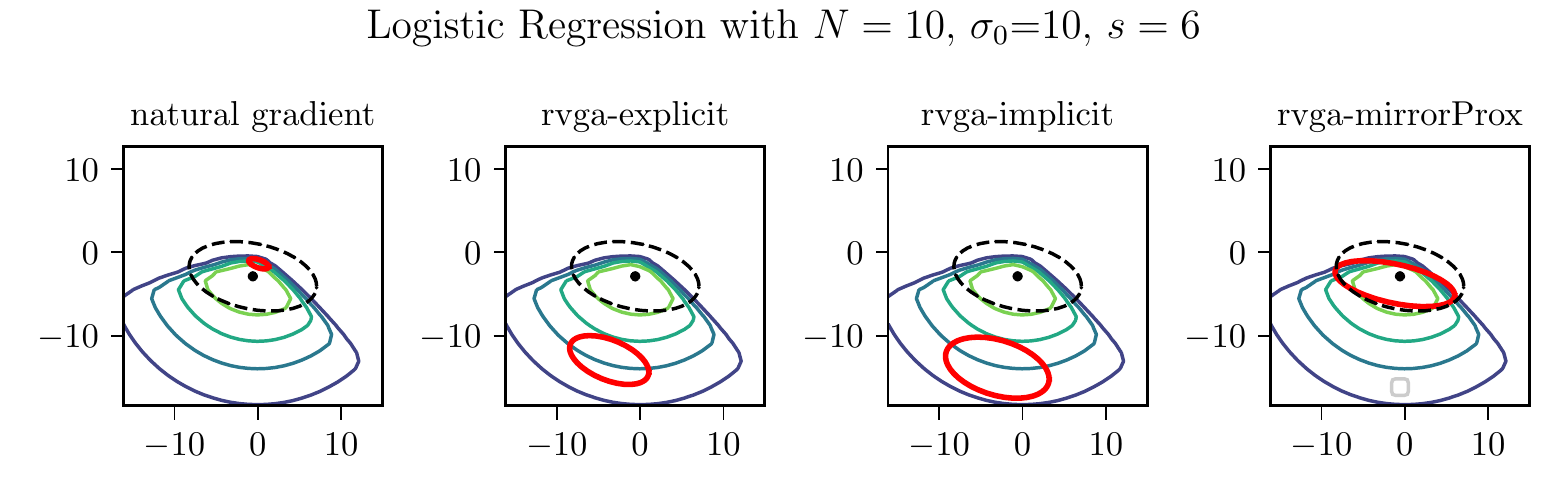}
\includegraphics[scale=1]{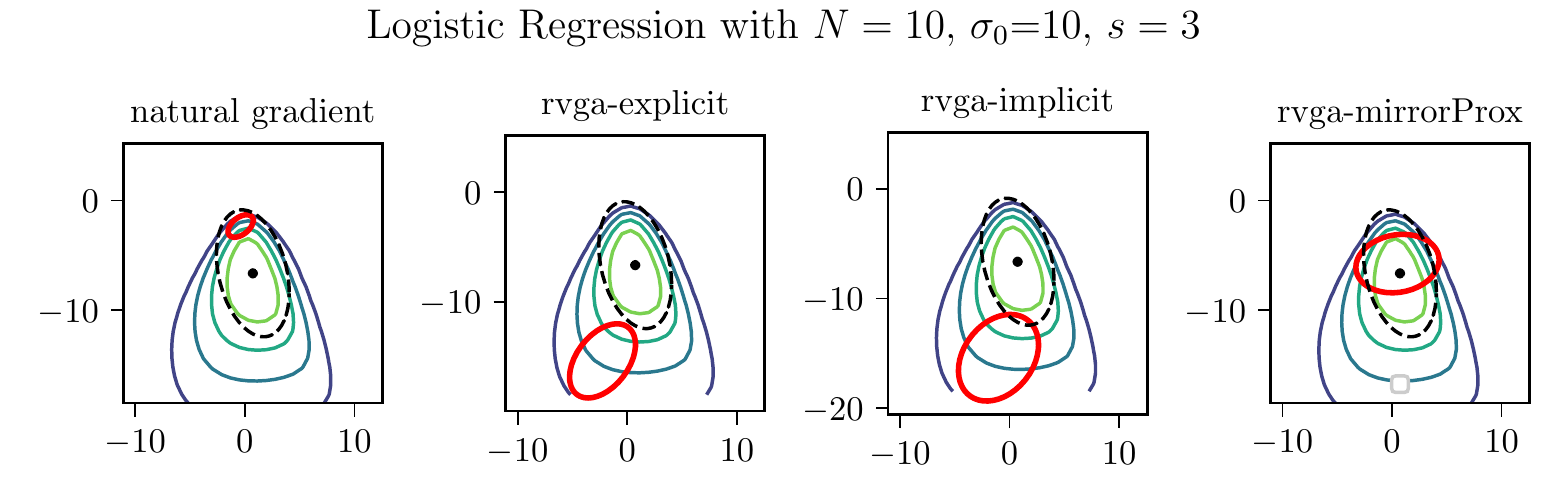}
\caption{Gaussian approximation of the Bayesian logistic posterior with a sharp prior $\sigma_0=10$ for different algorithms. The confidence ellipsoids of the Gaussians at the final time are shown in red. The contour lines of the true posterior are displayed in green. The batch Laplace ellipsoid is shown in a dashed line. We compare the mirror-prox updates \ref{MP-1}-\ref{MP-4}  (right column) with other variants of the updates: explicit using only \ref{MP-1}-\ref{MP-2} (second column),  implicit using \ref{updateMu}-\ref{updateP}  (third column). The natural gradient (left column) corresponds to a variant where the models are linearized. The mirror-prox schemes clearly better approximate the Bayesian logistic posterior.}
\label{figureXP_ExtraGrad}
\end{figure}

\end{document}